\documentclass[a4paper,UKenglish,cleveref, autoref, thm-restate]{lipics-v2021}

\usepackage[T1]{fontenc}
\usepackage{amsmath,amssymb}
\usepackage{enumitem}
\usepackage{amsmath}
\usepackage{xpatch}

\xapptocmd\normalsize{%
 \setlength\abovedisplayskip{5pt}
 \setlength\belowdisplayskip{5pt}
 \setlength\abovedisplayshortskip{0pt}
 \setlength\belowdisplayshortskip{0pt}
}{}{}

\usepackage{amsthm}
\makeatletter
\def\thm@space@setup{\thm@preskip=4pt \thm@postskip=4pt}
\makeatother

\allowdisplaybreaks
\newtheorem{postulate}{Postulate}{\bfseries}{\rmfamily}

\title{A Behavioural Theory of Probabilistic Algorithms Using Probabilistic Abstract State Machines}

\titlerunning{Behavioural Theory of Probabilistic Algorithms}

\author{Flavio Ferrarotti}{Software Competence Center Hagenberg, Austria}{johnqpublic@dummyuni.org}{https://orcid.org/0000-0003-2278-8233}{The research was supported by the Austrian ministries BMIMI, BMWET and the State of Upper Austria in the frame of the SCCH COMET module DEPS (FFG 888338) and the SCCH COMET competence center INTEGRATE (FFG 892418).}

\author{Klaus-Dieter Schewe}{Institut Nationale Polytechnique de Toulouse / IRIT CNRS, Toulouse, France}{kd.schewe@gmail.com}{https://orcid.org/0000-0002-8309-1803}{Part of the research has been supported by the ANR project EBRP: EventB-Rodin-Plus under grant no. ANR-19-CE25-0010.}

\authorrunning{F. Ferrarotti and K.-D. Schewe} 

\Copyright{Flavio Ferrarotti and Klaus-Dieter Schewe}



\category{} 

\relatedversion{} 

\nolinenumbers 

\begin{document}

\maketitle

\begin{abstract}
We motivate an axiomatic definition of probabilistic algorithms (PAs) by four postulates covering random branching time, abstract states, background, and random bounded exploration. Then, we introduce probabilistic Abstract State Machines (pASMs) and show that they specify PAs. Finally, we prove that every PA satisfying these postulates can be simulated step-by-step by a behaviourally equivalent pASM with the same signature and background.
\end{abstract}

\newpage
\section{Introduction}

\emph{Probabilistic algorithms} (PAs) are algorithms that involve some random choices so that the average expected outcome is better and more efficiently reachable than the outcome of either determining first the best choice or exploring all possible choices~\cite[p.233ff]{brassard:1988}. They are ubiquitous in modern computer science, with applications ranging from sorting and search~\cite{karp:dam1991} to cryptographic protocols~\cite{rabin:1980}, Monte Carlo methods in numerical computation and statistical physics~\cite{brassard:1988}, and machine learning through stochastic optimisation and genetic algorithms~\cite{holland:scholar2012}, to name just a few. The practical success of randomisation as an algorithmic design principle is beyond dispute. On the formal side, probabilistic Turing machines~\cite{Gill:1977} define the class of computable probabilistic functions, and significant work exists on rigorous specification and verification of probabilistic systems, including probabilistic program refinement~\cite{mciver:2005} and probabilistic model checking~\cite{baier:2008}.


What is nonetheless missing is a behavioural theory in the sense of Gurevich's programme~\cite{gurevich:lipari1995}: a model-independent axiomatic characterisation of probabilistic algorithms together with a proof that a natural abstract machine model captures exactly this class in the sense of step-by-step behavioural equivalence, that is, every algorithm satisfying the axioms can be simulated by the machine model so that each of its computation steps corresponds to a single step of the machine, producing the same observable behaviour (including the same possible outcomes and their probabilities), and vice versa. Existing formal approaches as those mentioned above 
are grounded in specific mathematical models rather than derived from first principles. A behavioural theory provides this foundation and underpins the ASM method for specification, verification, and refinement~\cite{boerger:2003,boerger:fac2003}. We develop such a theory within the framework of behavioural theories and Abstract State Machines (ASMs), briefly reviewed next.

\vspace*{0.2cm}
\noindent
\textbf{Behavioural Theories and the ASM Method.}
A \emph{behavioural theory} consists of a set of \emph{postulates} defining a class of algorithms, an \emph{abstract machine model} representing it, and a \emph{capture theorem} establishing that the two coincide via step-by-step behavioural equivalence. The ur-instance is Gurevich's theory of sequential algorithms~\cite{gurevich:tocl2000}, the \emph{sequential ASM thesis}, based on three postulates: \emph{sequential time}, \emph{abstract states}, and \emph{bounded exploration}, together with an implicit \emph{background} postulate providing Boolean logic and \textit{undef}. The corresponding model is the \emph{sequential ASM}, using assignments, conditionals, and bounded parallelism.

These ideas have been generalised to other classes, including parallel algorithms captured by parallel ASMs with unbounded parallelism~\cite{blass:tocl2003,blass:tocl2008}. A simplified axiomatisation with four postulates was later given by Ferrarotti et al.~\cite{ferrarotti:tcs2016}, using multiset comprehension terms in bounded exploration witnesses. Further behavioural theories cover recursive~\cite{boerger:fi2020}, concurrent~\cite{boerger:ai2016}, bulk synchronous parallel~\cite{ferrarotti:scp2019}, and reflective algorithms~\cite{schewe:scp2022,schewe:arxiv2025}.


The corresponding machine models support rigorous specification~\cite{boerger:2003,boerger:2018,boerger:2024}, refinement~\cite{boerger:fac2003}, and verification via dedicated logics. These include logics for deterministic ASMs~\cite{staerk:jucs2001}, general ASM logics for non-deterministic and concurrent systems~\cite{ferrarotti:igpl2017,ferrarotti:amai2018}, and extensions to reflective and temporal reasoning~\cite{schewe:scp2021,ferrarotti:abz2024}.


\vspace*{0.2cm}
\noindent
\textbf{Contribution.} We make two main contributions. The first, presented in Section~\ref{sec:postulates}, is an \emph{axiomatic characterisation of PAs} by four postulates. Identifying the right postulates is the central challenge: they must be simultaneously necessary, intuitively compelling, and strong enough to support a capture theorem. The postulate for \emph{random branching time} stipulates that for each state there exists a probability distribution over its possible successor states, reconceiving the very notion of a run as a sequence in which each successor state is drawn from that distribution. 
The postulates for \emph{abstract states} and \emph{background} both extend their counterparts in the parallel ASM thesis~\cite{ferrarotti:tcs2016}. The former is strengthened to require that isomorphisms between states preserve the probability distribution over their successors. The latter is extended by adding probability values, i.e.\ the set $[0,1] \subset \mathbb{R}$ equipped with multiplication, truncated addition, and ordering, as a background class for expressing and combining probabilities in state transitions.
The conceptually novel postulate is \emph{probabilistic bounded exploration}, whose key addition over its deterministic counterpart is a \emph{slicing condition}: sharpening the choice condition in a witness term yields a sub-algorithm with fewer possible runs but preserves the states that do occur, with probabilities renormalised by conditionalisation. This condition is what makes it possible to separate, within a single postulate, the combinatorial structure of updates from the probabilistic structure of choices, and identifying it is the conceptual core of the axiomatisation.

The second key contribution,  presented in Section~\ref{sec:thesis}, is the \emph{capture theorem}: every PA satisfying the four postulates is behaviourally equivalent to a \emph{probabilistic ASM} (pASM). The pASM model (Section~\ref{sec:pasm}) augments the standard parallel ASM model with a \emph{probabilistic choice rule}. This rule assigns a weight to each potential outcome and selects one based on its normalized probability (the relative likelihood of that outcome compared to all others). That pASMs satisfy the postulates is rather straightforward. The converse is the paper's main technical achievement. The proof reduces the problem to the parallel case by decomposing each step of a PA into two conceptual substeps: a parallel substep that determines, in parallel, all characterising tuples for possible successor states along with their probabilities, and a probabilistic substep that selects among them. This decomposition is non-obvious: it requires showing that the slicing condition guarantees the existence of characterising tuples, that these tuples are critical with respect to (probabilistic) witness terms, and that the complex isolating-formula machinery from the simplified parallel ASM thesis~\cite{ferrarotti:tcs2016} applies to both substeps. The resulting construction yields, for each state, a pASM rule capturing precisely the correct update sets and probability distribution; a structural similarity argument then allows us to collapse infinitely many state-specific rules into a single finite pASM rule, completing the proof.

\vspace*{0.2cm}
\noindent
\textbf{Scope and Related Work.} The theory covers standard classes of probabilistic algorithms~\cite{brassard:1988,karp:dam1991}, including  numerical, Las Vegas, Monte Carlo, and Sherwood algorithms, as well as genetic algorithms~\cite{holland:scholar2012}. Quantum algorithms~\cite{nielsen:2001} fall outside the present scope, as their randomness resides in states rather than transitions. 
There is an overlap with a behavioural theory of non-deterministic algorithms, which has been developed recently \cite{schewe:eb80}: Gurevich's sequential ASM thesis has been extended to capture choices among an a priori fixed set of alternatives that does not depend on the state~\cite{gurevich:tocl2000} (see also~\cite{boerger:2003}), and Schewe et al. further extend the theory to cover unbounded choice without unbounded parallelism, which was only sketched in~\cite{schewe:gki60-2016}. The present work is complementary to these developments. 
 On the verification side, a corresponding logic for pASMs is a natural next step not pursued~here.

\section{Axiomatisation of Probabilistic Algorithms} \label{sec:postulates}

We start with motivating and defining postulates that define the class of probabilistic algorithms (PAs). As already indicated above we will emphasise probabilistic choices in state transitions, but disregard randomised states. This is in accordance with the kind of probabilistic algorithms treated by Brassard and Bratley \cite[Chap.8]{brassard:1988} and Karp \cite{karp:dam1991}. It also covers randomness as in genetic algorithms \cite{koza:1993}, but it excludes quantum algorithms, where states themselves need to be considered as being randomised.

\subsection{Probabilistic Branching Time}
\label{ssec:rbt}

In all existing behavioural theories for sequential, parallel, concurrent, and other classes of algorithms, an algorithm $\mathcal{A}$ comprises a set $\mathcal{S}$ of states together with a distinguished subset $\mathcal{I} \subseteq \mathcal{S}$ of initial states. This remains the same for PAs. What changes is the nature of state transitions. Because a PA involves random choices, there is no single successor state; instead we have a \emph{state transition relation} $\tau \subseteq \mathcal{S} \times \mathcal{S}$, and we write $\tau(S) = \{ S^\prime \mid (S, S^\prime) \in \tau \}$ for the set of \emph{successor states} of $S$.

The successor states in $\tau(S)$ arise from a random choice made at $S$, so $\tau(S)$ must carry a probability distribution. 
In a single computational step, a probabilistic algorithm selects among a set of alternatives that it can explicitly enumerate at the current state; as only finitely many alternatives can be considered in one step\footnote{While probabilistic programs abstractly utilize continuous distributions, their operational semantics on digital hardware are constrained to finite floating-point or discrete machine representations. A state-dependent, finite-but-unbounded choice space thus precisely models the concrete transition systems of these algorithms, as uncountably infinite branching cannot be implemented at machine level}, we assume that $\tau(S)$ is finite for every state $S$. A probability measure on a finite set requires no measure-theoretic apparatus: it is simply a function $d_S \colon \mathcal{P}(\tau(S)) \to [0,1]$ with $d_S(\tau(S)) = 1$ and $d_S\!\left(\bigcup_{i \in I} X_i\right) = \sum_{i \in
I} d_S(X_i)$ for any family $\{X_i\}_{i \in I}$ of pairwise disjoint sets $X_i \subseteq \tau(S)$.

\begin{postulate}[Random Branching Time]\label{p-rbt}
A probabilistic algorithm $\mathcal{A}$ comprises a set of states $\mathcal{S}$, a set of initial states $\mathcal{I} \subseteq \mathcal{S}$ and for each state $S$ a probability measure $d_S : \mathcal{S} \rightarrow [0,1]$.
\end{postulate}

 A state $S^\prime$ with $d_S(S^\prime) > 0$ is called a \emph{possible successor state} of $S$. Then for any $X \subseteq \mathcal{S}$ we have $d_S(X) = \sum_{S^\prime \in X} d_S(S^\prime)$, which gives the probability of $X$ containing a possible successor state of $S$. In particular, $d_S(\mathcal{S}) = 1$.
 
 With Postulate \ref{p-rbt} we already get the definition of a {\em run} of a PA $\mathcal{A}$ as a sequence of states $S_0, S_1, \dots$ with $S_0 \in \mathcal{I}$ and $S_i \in \mathcal{S}$ with $d_{S_i}(S_{i+1}) > 0$ for all $i \ge 0$. The probability of such a run is $\prod_{i \geq 0} d_{S_i}(S_{i+1})$, which is well-defined, as the sequence of partial products $\{ \prod_{i=0}^k d_{S_i}(S_{i+1}) \}_{k \ge 0}$ is monotone decreasing. 

 As usual in behavioural theories, algorithms are treated as abstract objects defined by their state space and transitions. Behavioural equivalence therefore amounts to equality of the induced transition behaviour on states.
 Two PAs $\mathcal{A}$ and $\mathcal{A}^\prime$ are \emph{behaviourally equivalent} if they have the same runs and their probability measures coincide, i.e.\ $d_S(S^\prime) = d^\prime_S(S^\prime)$ holds for all states $S, S^\prime \in \mathcal{S}$.

\subsection{Abstract States of Probabilistic Algorithms}
\label{ssec:abstract-states}

As we consider only random choices in state transitions, but not random states, we can preserve the notion of state from the parallel ASM thesis \cite{blass:tocl2003,ferrarotti:tcs2016}. For this consider a {\em signature} $\Upsilon$, i.e. a finite set of function symbols, where each $f \in \Upsilon$ has some fixed arity $n_f$. Given a set $B$, we can define a {\em structure} $S$ over $\Upsilon$ as a family of functions $f_S: B^{n_f} \rightarrow B$ (for all $f \in \Upsilon$). Then $B$ is called the {\em base set} of the structure $S$. In order to permit also partial functions, we can simply assume the presence of a special value $\textit{undef\/} \in B$ and request $f_S(b_1, \dots, b_{n_f}) = \textit{undef\/}$, whenever one of the $b_i$ is \textit{undef\/}.

An {\em isomorphism} between $\Upsilon$-structures $S$ and $S^\prime$ is a bijection $\sigma: B \rightarrow B^\prime$ between the corresponding base sets with $\sigma(\textit{undef\/}) = \textit{undef\/}$ such that $f_{S^\prime}(\sigma(b_1), \dots, \sigma(b_{n_f})) = \sigma(f_S(b_1, \dots, b_{n_f}))$ holds for all $b_1, \dots, b_{n_f} \in B$. With this we can formulate the following second postulate for probabilistic algorithms.

\begin{postulate}[Abstract State]\label{p-state}
There exists a signature $\Upsilon$ such that all states $S \in \mathcal{S}$ of a probabilistic algorithm $\mathcal{A}$ are $\Upsilon$-structures, and the sets $\mathcal{S}$, $\mathcal{I}$ are closed under isomorphisms. Whenever $d_S(S^\prime) > 0$ holds, the states $S$ and $S^\prime$ have the same base set. For any isomorphism $\sigma$ defined on state $S$ we get $d_{\sigma(S')}(\sigma(S^\prime)) = d_S(S^\prime)$.
\end{postulate}

With Postulate \ref{p-state} we already get the means to analyse differences between states, for which we use the notion of update sets. For a function symbol $f \in \Upsilon$ of arity $n$ and values $b_1, \dots, b_n \in B$ for the base set $B$ of some state $S$ we call the pair $\ell = (f,(b_1, \dots, b_n))$ a {\em location} of $S$, and $b_0 = f_S(b_1, \dots, b_n)$ is called the {\em value} of this location in $S$, which we also denote as $\text{val}_S(\ell)$. An {\em update} in $S$ is a pair $(\ell, v)$ with a location $\ell$ of $S$ and a value $v \in B$. 
An {\em update set} $\Delta$ in $S$ is a set of such updates and it is called {\em consistent} iff for all $(\ell, v_1), (\ell, v_2) \in \Delta$ we have $v_1 = v_2$, i.e., there can be at most one update value for any location. For consistent $\Delta$ we define the state $S^\prime = S + \Delta$ by $\text{val}_{S^\prime}(\ell) = v$ for $(\ell, v) \in \Delta$ and $\text{val}_{S^\prime}(\ell) = \text{val}_S(\ell)$ for all other locations. We extend this definition by $S + \Delta = S$ for inconsistent $\Delta$.

\begin{lemma}\label{lem-update-set}
For any two states $S, S^\prime$ with the same base set $B$ there exists a unique minimal consistent update set $\Delta = \Delta(S,S^\prime)$ with $S^\prime = S + \Delta$.
\end{lemma}
\vspace{-1em}
\begin{proof}
Let $\mathcal{D} = \{ \ell \mid \text{val}_{S^\prime}(\ell) \neq \text{val}_S(\ell) \}$ be the set of locations, at which the two states differ, and define the update set $\Delta = \{ (\ell, \text{val}_{S^\prime}(\ell)) \mid \ell \in \mathcal{D} \}$. Clearly, $\Delta$ is consistent, and we get $S + \Delta = S^\prime$. The minimality is a consequence of the definition of $\mathcal{D}$.
\end{proof}

\vspace{-0.5em}
As we are dealing with algorithms, we must ensure that update sets are always finite. As discussed in the context of the logic of parallel ASMs \cite{ferrarotti:igpl2017} this can be formally achieved by exploiting meta-finite structures \cite{graedel:infcomp1998}. With finite update sets it will always be the case that almost all locations in states of some run have the same value as in the initial state of that run. Therefore, without loss of generality we can assume that in every state there are only finitely many locations that have a value different from \textit{undef\/}.

For the remainder of this article let $\boldsymbol{\Delta}(S) = \{ \Delta(S, S^\prime) \mid d_S(S^\prime) > 0 \}$ denote the \emph{set of update sets} on the state $S$. Each $\Delta \in \boldsymbol{\Delta}(S)$ defines a possible successor state $S^\prime = S + \Delta$ of the state $S$ and vice versa. The probability measure $d_S$ on states thereby induces a probability measure on $\boldsymbol{\Delta}(S)$, which by a slight abuse of notation we also denote $d_S$, defined by $d_S(\Delta(S,S^\prime)) = d_S(S^\prime)$.



\subsection{Background of Probabilistic Algorithms}

In every algorithm we make implicit assumptions about values in base sets and functions on them that are available in every state. In existing behavioural theories we refer to this as the \emph{background}. For instance, we already mentioned the value $\textit{undef\/}$. Backgrounds in general have been investigated in depth by Blass and Gurevich \cite{blass:beatcs2007}. For PAs we extend the background requirements used for synchronous parallel algorithms~\cite{ferrarotti:tcs2016} by adding probability values as a background class, which are needed to express and combine probabilities in state transitions. This leads to the following third postulate.


\begin{postulate}[Background]\label{p-background}
Each state contains:

\begin{itemize}[noitemsep, topsep=0pt, parsep=0pt, partopsep=0pt]

\item truth values \textbf{true} and \textbf{false} and the Boolean operations on them;

\item a special value \textit{undef\/};

\item the equality predicate $=$;

\item all ordered pairs as well as binary function symbols $(x, y)$
for pairing and unary projection functions $\pi_1, \pi_2$ for extracting the first and second components of a pair;

\item all finite multisets together the empty multiset $\langle\rangle$, singletons $\langle x \rangle$, binary sum $x \oplus y$, and general sum $\bigoplus x$ as well as functions \textit{TheUnique\/} and \textit{AsSet\/}.

\item all probability values, i.e.\ the set $[0,1] \subset \mathbb{R}$, together with multiplication $\times$, truncated addition $x +_{[0,1]} y = \min(x + y, 1)$, the standard ordering $\leq$, and the distinguished elements $0$ and $1$.

\end{itemize}
\end{postulate}

In a multiset $x \oplus y$ the multiplicity of a value $v$ is the sum of the multiplicities of $v$ in $x$ and $y$, respectively. Analogously, the multiplicity of a value $v$ in $\bigoplus x$ is the sum of all multiplicities in the elements of $x$. We have $\textit{TheUnique\/}(\langle x \rangle) = x$ and $\textit{TheUnique\/}(x) = \textit{undef\/}$ for all $x$ that are not singleton multisets. $\textit{AsSet\/}(x)$ is the multiset that has the same elements as the multiset $x$, but all with multiplicity 1 (so this multiset is in fact a set). Multiplication $\times$ is closed on $[0,1]$ in the usual sense. Truncated addition $x +_{[0,1]} y = \min(x + y, 1)$ ensures closure of addition on $[0,1]$; it is used when combining probabilities that are guaranteed by construction to sum to at most $1$, in which case it coincides with standard addition.

\subsection{Probabilistic Bounded Exploration}

The decisive observation underlying the sequential ASM thesis~\cite{gurevich:tocl2000} as well as the (simplified) parallel ASM thesis \cite{ferrarotti:tcs2016} is that the behaviour of a sequential or parallel algorithm, respectively, is completely determined by some finite specification. This specification must contain terms defined over the signature $\Upsilon$, which are evaluated in a state in order to build the successor state. Hence there must exist a finite set of terms---ground terms in the case of the sequential thesis, multiset comprehension terms in case of the parallel thesis---that determine the update set, the application of it will result in the successor state.

For PAs the same argument applies with the difference that there are finitely many successor states, which are associated with probabilities. The evaluation of multiset comprehension terms in a state gives rise to a multiset, and the varying size of these can be exploited to capture unbounded parallelism, where the number of parallel branches depends on the state. Then multiset comprehension terms could also capture the varying number of possible successor states. We therefore define {\em witness terms} as a double multiset comprehension terms of the form
\begin{equation}\label{eq-witness}
\alpha = \langle U(y_1, \dots, y_q) \mid \psi(y_1, \dots, y_q) \rangle \; ,     
\end{equation}
where $U(y_1, \dots, y_q)$ is a multiset comprehension term of the form
\begin{equation*}
\beta = \langle (t_0, \dots, t_n) \mid \varphi(x_1, \dots, x_r, y_1, \dots, y_q) \rangle \; ,     
\end{equation*}
where the $x_i$ and $y_j$ are the only variables appearing in the terms $t_k$. 

Informally, the inner condition $\varphi$ is meant to capture parallelism, whereas the outer condition $\psi$ is to capture choices. In order to distinguish these two roles, we consider {\em sliced witness terms} of the form 
\begin{equation}\label{eq-witness-mod}
\alpha^\prime = \langle U(y_1, \dots, y_q) \mid \psi(y_1, \dots, y_q) \wedge \chi(y_1, \dots, y_q) \rangle
\end{equation}
with some additional condition $\chi$ for a witness term $\alpha$ as in (\ref{eq-witness}). 

In addition, as we also have to take probabilities of choices into account, we further consider \emph{p-witness terms}, which are multiset comprehension terms of the form $\varrho = \langle (y_1, \dots, y_q, p(y_1, \dots, y_q)) \mid \psi(y_1, \dots, y_q) \rangle$, where the term $p$ evaluates to values in $[0,1]$. We further require that in every state the values produced by a p-witness term form a valid probability distribution over the choices identified by $\psi$, i.e.\ $\sum_{\bar{b}:\, \psi(\bar{b}) = \mathbf{true}} p(\bar{b}) = 1$. We say that a witness term $\alpha$ is \emph{aligned} with a p-witness term $\varrho$ iff the defining condition $\psi(y_1, \dots, y_q)$ is the same for both. A set
$W$ of witness terms is \emph{aligned} with a set $P$ of p-witness terms iff every $\alpha \in W$ is aligned to some $\varrho \in P$ and vice versa.
With this we get the following fourth postulate for PAs. 


\begin{postulate}[Probabilistic Bounded Exploration]\label{p-bew}
For a probabilistic algorithm $\mathcal{A}$ there exists a finite set $W$ of witness terms that is aligned to a finite set $P$ of p-witness terms such that for all states $S, S^\prime$ that coincide on $W \cup P$, i.e. $\text{val}_S(\alpha) = \text{val}_{S^\prime}(\alpha)$ holds for all $\alpha \in W$ and $\text{val}_S(\varrho) = \text{val}_{S^\prime}(\varrho)$ holds for all $\varrho \in P$, we have $\boldsymbol{\Delta}(S) = \boldsymbol{\Delta}(S^\prime)$ and $d_S = d_{S^\prime}$, where $d_S$, $d_{S^\prime}$ refer to the probability measures on $\boldsymbol{\Delta}(S)$ and $\boldsymbol{\Delta}(S')$, respectively.

In addition, if $W$ is replaced by $W^\prime$ containing sliced witness terms $\alpha^\prime$ for all $\alpha \in W$, then this defines a probabilistic algorithm $\mathcal{A}^\prime$ with the same signature and background as $\mathcal{A}$ such that $\boldsymbol{\Delta}^{\mathcal{A}^\prime}(S) \subseteq \boldsymbol{\Delta}^{\mathcal{A}}(S)$ holds for all states $S$, and $d_S^{\mathcal{A}^\prime}(S^\prime) = d_S^{\mathcal{A}}(S^\prime) / \sum_{S^{\prime\prime} \in \tau^\prime(S)} d_S^{\mathcal{A}}(S^{\prime\prime})$ for $S^\prime \in \tau^\prime(S)$ and $d_S^{\mathcal{A}^\prime}(S^\prime) = 0$ otherwise, where $\tau^\prime(S) = \{ S + \Delta \mid \Delta \in \boldsymbol{\Delta}^{\mathcal{A}^\prime}(S) \}$.

\end{postulate}

Here we use superscripts $\mathcal{A}$ and $\mathcal{A}^\prime$, respectively, to distinguish between sets of update sets and probability measure associated with the PAs $\mathcal{A}$ and $\mathcal{A}^\prime$. The set $W \cup P$ of witness and p-witness terms will be called a {\em bounded exploration witness}. Note that the slicing preserves the states, but reduces the possible successor states and consequently also the possible runs, which is what we should expect, if we only restrict choices. 

Furthermore, the modification of the probability measure uses conditional probabilities for the condition $S^\prime \in \tau^\prime(S)$. This implies that the modified algorithm $\mathcal{A}^\prime$ has a bounded exploration witness $W^\prime \cup P^\prime$, where each $\varrho^\prime \in P^\prime$ has the form $\varrho^\prime = \langle (y_1, \dots, y_q, p^\prime(y_1, \dots, y_q)) \mid \psi(y_1, \dots, y_q) \rangle$ and is derived from some $\varrho = \langle (y_1, \dots, y_q, p(y_1, \dots, y_q)) \mid \psi(y_1, \dots, y_q) \rangle \in P$.

\subsection{Example: Probabilistic Turing Machines as PAs}
\label{ssec:ptm}

A natural question is whether the Postulates~\ref{p-rbt}--\ref{p-bew} are general enough to capture standard models of probabilistic computation. To explore this, we consider probabilistic Turing machines (PTMs), which extend ordinary Turing machines by allowing transitions to be governed by probability distributions. 

A PTM is a tuple $
M = (Q, \Sigma, \Gamma, \delta, q_0, q_{\mathrm{acc}}, q_{\mathrm{rej}})$, where $Q$ is a finite set of states, $\Sigma \subseteq \Gamma$ are the input and tape alphabets, $q_0$ is the initial state, $q_{\mathrm{acc}}$ and $q_{\mathrm{rej}}$ are accepting and rejecting states, and $\delta : Q \times \Gamma \to \mathcal{D}(Q \times \Gamma \times \{L,R\})$ with $\mathcal{D}(X)$ the set of discrete probability distributions on $X$: for $p \in \mathcal{D}(X)$ we have $p(x) \in [0,1]$ for $x \in X$ and $\sum_{x \in X} p(x) = 1$. Notice that $Q$ and $\Gamma$ are finite, hence only finitely many $x$ have $p(x) > 0$. A configuration is a triple $(q, w, h)$ with current state $q$, tape $w:\mathbb{Z}\to\Gamma$, and head position $h$.

Let $\mathcal{S}$ be the set of all $\Upsilon$-structures encoding configurations of a PTM $M$. Each configuration $C = (q,w,h)$ of $M$ induces a discrete probability measure $d_C$ over successor configurations via $\delta(q,w(h))$, satisfying Postulate~\ref{p-rbt}.
Encoding $C$ as a $\Upsilon$-structure with the tape as a unary function and constants for head and control state, with the base set containing tape positions, symbols, and control states, ensures isomorphisms correspond to renaming tape positions and preserve $d_C$, satisfying Postulate~\ref{p-state} (cf. Gurevich~\cite[Sec.~4.3.1]{gurevich:tocl2000}).
The background, including tape indices, symbols, Booleans, equality, arithmetic on indices and probability values, is fixed in all states, satisfying Postulate~\ref{p-background}.

Let $\mathit{prov} : Q \times \Gamma \times Q \times \Gamma \times \{L,R\} \rightarrow [0,1]$ such that $\mathit{prov}(q,a,q',a',m) = p(q',a',m)$ for $\delta(q,a) = p$. We define the following witness terms
\begin{align*}
\alpha_1 &= \langle \langle (x) \mid \mathbf{true} \rangle \mid \mathit{prov}(q, w(h), x, y, z) > 0 \rangle \\
\alpha_2 &= \langle \langle (y,h) \mid \mathbf{true} \rangle \mid \mathit{prov}(q, w(h), x, y, z) > 0 \rangle \\
\alpha_3 &= \langle \langle (z = L) \mid \mathbf{true} \rangle \mid \mathit{prov}(q, w(h), x, y, z) > 0 \rangle \\
\alpha_4 &= \langle \langle (h-1) \mid z = L \rangle \mid \mathit{prov}(q, w(h), x, y, z) > 0 \rangle \\
\alpha_5 &= \langle \langle (z = R) \mid \mathbf{true} \rangle \mid \mathit{prov}(q, w(h), x, y, z) > 0 \rangle \\
\alpha_6 &= \langle \langle (h+1) \mid z = R \rangle \mid \mathit{prov}(q, w(h), x, y, z) > 0 \rangle 
\end{align*}
and the associated p-witness term 
\[\varrho = \langle (q, w(h),x,y,z,\mathit{prov}(q,w(h),x,y,z)) \mid \mathit{prov}(q, w(h), x, y, z) > 0 \rangle. \]

Next, we show that the witness set $W = \{ \alpha_i \mid 1 \le i \le 6 \}$ and the set of p-witness terms $P = \{\varrho\}$ satisfy Postulate~\ref{p-bew} for PTMs.

Let $S,S'\in\mathcal{S}$ be $\Upsilon$-states encoding configurations $C=(q,w,h)$ and $C'=(q',w',h')$ of a PTM $M$.  If $\mathrm{val}_S(\alpha_i)=\mathrm{val}_{S'}(\alpha_i)$ for all $1 \le i \le 6$ and $\mathrm{val}_S(\varrho)=\mathrm{val}_{S'}(\varrho)$ hold, then necessarily the interpretation of $q$ and $w(h)$ coincide in $S$ and $S'$, and the same set of triples $(x,y,z)$ with $\mathit{prov}(q,w(h),x,y,z)>0$ is obtained in both states with identical probabilities. Each such triple determines an update set $\Delta_{(x,y,z)}$ which updates $q$, $w(h)$ and $h$, respectively, to the values $x$, $y$, and $h-1$ if $z = L$ or $h+1$ if $z = R$.
Hence $\boldsymbol{\Delta}(S)=\boldsymbol{\Delta}(S')$ and the induced probability measures coincide, $d_S(\Delta_{(x,y,z)}) = d_{S'}(\Delta_{(x,y,z)}) = \mathit{prov}(q,w(h),x,y,z)$. 

Thus, the finite sets $W=\{\alpha_1 ,\dots, \alpha_6\}$ and $P=\{\varrho\}$ form a bounded exploration witness for $M$ and show that the PTM $M$ satisfies Postulate~\ref{p-bew}. 
Moreover, slicing $\alpha$ to obtain $\alpha'$ corresponds to removing some of the admissible transition choices.  
This clearly yields again a PTM, but with fewer possible runs, since certain transitions have been eliminated.  
The probability distribution on the remaining runs is obtained by renormalising the original probabilities via conditionalisation, 
exactly as required by Postulate~\ref{p-bew}.

\section{Probabilistic Abstract State Machines}\label{sec:pasm}

We proceed with the definition of an abstract machine model. For this we first need to define states and their background, then continue with probabilistic state transitions.

\subsection{States of pASMs}

States and background of a \emph{probabilistic Abstract State Machine} (pASM) $M$ are defined as for parallel ASMs~\cite{ferrarotti:tcs2016} with the extension by probability values, same as in Postulate \ref{p-background}. That is, states result from the interpretation of a signature $\Upsilon$ over a base
set $B$, which must contain truth values \textbf{true} and \textbf{false}, probability values in [0,1], a special value $\textit{undef\/}$, and be closed under finite multiset construction and pairing. The background further contains the equality predicate $=$, the Boolean operators, basic arithmetic on integers, projection operators for pairs, and the multiset operators $\oplus$ for binary sum, $\bigoplus$ for general sum, \textit{TheUnique\/} and \textit{AsSet\/}. The sets $\mathcal{S}$ of all states and a specified subset $\mathcal{I}$ of initial states of $M$ are both closed under isomorphisms.

The single extension over parallel ASM states and background is that each base set $B$ must additionally contain all probability values $[0,1] \subset \mathbb{R}$ together with multiplication $\times$, truncated addition $x +_{[0,1]} y = \min(x+y,1)$, and the standard ordering $\leq$ on $[0,1]$.

The signature $\Upsilon$ of $M$, the background of $M$, and an arbitrary set $V$ of variables define the set of \emph{terms} of $M$ in the usual way. This includes tuple terms defined as shortcuts for nested pairing and multiset comprehension terms $\langle t(x_1, \dots, x_n) \mid \varphi(x_1, \dots, x_n) \rangle$ with variables $x_i$ and terms $t$ and $\varphi$; the latter should be Boolean, i.e.\ its evaluation in a state with respect to a given variable assignment should result in a truth value. Let $\mathbb{T}_M(V)$ denote the set of terms of $M$ with variables in $V$.

\subsection{Rules of pASMs}
\label{ssec:rules}

State transitions of $M$ are defined by means of pASM rules as follows.

\begin{itemize}[noitemsep, topsep=0pt, parsep=0pt, partopsep=0pt]

\item \textbf{skip} is a rule, called {\em skip rule}.

\item For $f \in \Upsilon$ and terms $t_0, t_1, \dots, t_{n_f}$ the expression $f(t_1, \dots, t_{n_f}) := t_0$ is a rule, called {\em assignment rule}.

\item For a Boolean term $\varphi$ and pASM rules $r_1$ and $r_2$ the expression $\textbf{if} \; \varphi \; \textbf{then} \; r_1 \; \textbf{else} \; r_2$ is a rule, called {\em conditional rule}.

\item For a Boolean term $\varphi$ with free variables $x_1, \dots, x_k$ and a pASM rule $r$ with free variables among the $x_i$ the expression $\textbf{forall} \; x_1, \dots, x_k \; \textbf{:} \; \varphi(x_1, \dots, x_k) \; \textbf{do} \; r(x_1, \dots, x_k)$ is a rule, called {\em parallel rule}.


\item For a Boolean term $\psi$ with free variables $x_1, \dots, x_k$, a non-negative real-valued term $p$ with free variables $x_1, \dots, x_k$ serving as a \emph{weight}, and a pASM rule $r$ with free variables among $x_1, \dots, x_k$, the expression \[\textbf{choose} \; x_1, \dots, x_k \; \textbf{:} \; \psi(x_1, \dots, x_k) \; \textbf{with weight} \; p(x_1, \dots, x_k) \; \textbf{do} \; r(x_1, \dots, x_k)\] is a rule, called \emph{probabilistic choice rule}. The actual probability assigned to each alternative is obtained by normalising all weights over the alternatives selected by $\psi$, as defined in the semantics below.

\end{itemize}

In order to avoid ambiguities the free use of parentheses is permitted. Furthermore, common shortcuts can be used, e.g. omitting \textbf{else} $r_2$ for the case of a skip rule $r_2$. The set of {\em free variables} in a rule is defined as usual.

The semantics of pASM rules is defined by a set of update sets $\boldsymbol{\Delta}_{r,\sigma}(S)$ for a state $S$ with base set $B$ and a variable assignment $\sigma: V \rightarrow B$ as well as an associated probability measure $d_{r,S,\sigma}$.

\begin{itemize}[noitemsep, topsep=0pt, parsep=0pt, partopsep=0pt]

\item For a skip rule we have $\boldsymbol{\Delta}_{\textbf{\textit{skip\/}},\sigma}(S) = \{ \emptyset \}$ and $d_{\textbf{\textit{skip\/}},S,\sigma}(\emptyset) = 1$.

\item For an assignment rule $f(t_1, \dots, t_{n_f}) := t_0$ we have $\boldsymbol{\Delta}_{r,\sigma}(S) = \{ \Delta \}$ with\\ $\Delta = \{ ((f, ( \text{val}_{S,\sigma}(t_1), \dots, \text{val}_{S,\sigma}(t_{n_f}) )), \text{val}_{S,\sigma}(t_0) ) \}$ and $d_{r,S,\sigma}(\Delta) = 1$.

\item For a conditional rule $\textbf{if} \; \varphi \; \textbf{then} \; r_1 \; \textbf{else} \; r_2$ we have $\boldsymbol{\Delta}_{r,\sigma}(S) = \boldsymbol{\Delta}_{r_1,\sigma}(S)$ and $d_{r,S,\sigma}(\Delta) = d_{r_1,S,\sigma}(\Delta)$ in case $\text{val}_{S,\sigma}(\varphi) = \textbf{true}$ holds, and $\boldsymbol{\Delta}_{r,\sigma}(S) = \boldsymbol{\Delta}_{r_2,\sigma}(S)$ and $d_{r,S,\sigma}(\Delta) = d_{r_2,S,\sigma}(\Delta)$ otherwise.


\item For a parallel rule $\textbf{forall} \; \bar{x} \; \textbf{:} \;
\varphi(\bar{x}) \; \textbf{do} \; r^\prime(\bar{x})$, where
$\bar x = (x_1,\ldots,x_k)$ and $\bar a = (a_1,\ldots,a_k)$ ranges over tuples of values in
$B$, let
\[
B_{S,\sigma} \;=\; \{\, \bar a \mid val_{S,\sigma[\bar x\mapsto\bar a]}(\varphi) = true \,\},
\]
which we assume to be finite (as is standard for parallel ASMs; see formal definition using meta-finite states in~\cite{ferrarotti:tcs2016}),  and let
\[
\Omega_{r,\sigma}(S) \;=\; \prod_{\bar a \in B_{S,\sigma}} \Delta_{r',\sigma[\bar x\mapsto\bar a]}(S)
\]
be the (finite) set of families $(\Delta_{\bar a})_{\bar a \in B_{S,\sigma}}$ of independent
local choices. We define
\[
\Delta_{r,\sigma}(S) \;=\; \Big\{\, \textstyle\bigcup_{\bar a} \Delta_{\bar a} \ \Big|\ (\Delta_{\bar a})_{\bar a} \in \Omega_{r,\sigma}(S) \,\Big\}
\]
\[
d_{r, S, \sigma}(\Delta) \;=\; \sum_{\substack{(\Delta_{\bar a})_{\bar a} \,\in\, \Omega_{r,\sigma}(S) \\ \bigcup_{\bar a} \Delta_{\bar a} \,=\, \Delta}}
\ \prod_{\bar a \in B_{S,\sigma}} d_{r',S,\sigma[\bar x\mapsto\bar a]}(\Delta_{\bar a})
\quad (\Delta \in \Delta_{r,\sigma}(S)).
\]

\item For a probabilistic choice rule $\textbf{choose} \; \bar{x} \; \textbf{:} \; \psi(\bar{x}) \; \textbf{with weight} \; p(\bar{x}) \; \textbf{do} \; r^\prime(\bar{x})$, where $\bar{x} = (x_1, \dots, x_k)$ and $\bar{a} = (a_1, \dots, a_k)$ ranges over tuples of values in $B$, let
\[
B_{S,\sigma} \;=\; \big\{\, \bar a \mid \text{val}_{S,\sigma[\bar x\mapsto\bar a]}(\psi) = true
\text{ and } \text{val}_{S,\sigma[\bar x\mapsto\bar a]}(p) > 0 \,\big\},\]
\[
\Sigma_{S,\sigma} \;=\; \sum_{\bar a \in B_{S,\sigma}} \text{val}_{S,\sigma[\bar x\mapsto\bar a]}(p) ,
\]
where $\Sigma_{S,\sigma}$ is a finite sum, as there are only finitely many non-zero values of
$p$. If $\Sigma_{S,\sigma} = 0$, set $\Delta_{r,\sigma}(S) = \{\emptyset\}$ and
$d_{r,\sigma}(\emptyset) = 1$ (i.e.\ $r$ behaves as $\textbf{skip}$). Otherwise we define
\[
\Delta_{r,\sigma}(S) \;=\; \bigcup_{\bar a \in B_{S,\sigma}} \Delta_{r',\sigma[\bar x\mapsto\bar a]}(S)
\]
\[
d_{r, S, \sigma}(\Delta) \;=\; \sum_{\substack{\bar a \,\in\, B_{S,\sigma} \\ \Delta \,\in\, \Delta_{r',\sigma[\bar x\mapsto\bar a]}(S)}}
\frac{\text{val}_{S,\sigma[\bar x\mapsto\bar a]}(p)}{\Sigma_{S,\sigma}} \cdot d_{r',S,\sigma[\bar x\mapsto\bar a]}(\Delta)
\quad (\Delta \in \Delta_{r,\sigma}(S)).
\]

\end{itemize}

Finally, a pASM $M$ comprises a background as defined above, a signature $\Upsilon$ defining the set $\mathcal{S}$ of states and the subset $\mathcal{I}$ of initial states, and a closed pASM rule $r$. This defines the state transition relation $\tau_M$ of $M$ with $(S, S^\prime) \in \tau_M$ iff there exists some update set $\Delta \in \boldsymbol{\Delta}_r(S)$ with $S + \Delta = S^\prime$. Furthermore, we get probability measures $d_S$ for all states $S$ of $M$, which are defined by $d_S(S + \Delta) = d_{r,S}(\Delta)$. Note that we can omit the variable assignment $\sigma$, because the rule $r$ is closed.

\subsection{Examples}

Brassard and Bratley~\cite[Chapter~8]{brassard:1988} classify probabilistic algorithms into four categories--\emph{numerical}, \emph{Las Vegas}, \emph{Monte Carlo}, and \emph{Sherwood}--according to the role that randomness plays in their computation and in the correctness of their results. 
We illustrate each class with a canonical algorithm, expressed as a pASM.  

Note that a probabilistic choice rule with weight $p = 1$ yields a uniform distribution over all alternatives selected by $\psi$: since all weights are equal, the normalisation by $\Sigma_{S,\sigma}$ assigns probability $1/|\{\bar{a} \mid \text{val}_{S,\sigma[\bar{x}\mapsto\bar{a}]}(\psi) = \mathbf{true}\}|$ to each alternative. This is used in the examples below to model uniform random choices.
%
%
The examples also use the standard ASM $\textbf{let } x = t \textbf{ in } r$ rule, which evaluates term $t$ in the current state, binds the result to variable $x$, and then executes rule $r$ with that binding; it can always be replaced by substituting $t$ for $x$ directly in $r$, but improves readability. As usual, two or more rules written in sequence are executed \emph{in parallel}: for rules $r_1, \dots, r_k$ this can be expressed by $\textbf{forall } i \in \{1,\dots,k\} \textbf{ do } r_i$.\\[0.1cm]
\noindent
\textbf{Numerical.} The \emph{Monte Carlo integration} method estimates $\int_0^1 f(x)\,dx$ by sampling $N$ random points $(x,y)$ uniformly from a finite grid $G \times G$, where $G = \{ k/M \mid 0 \le k \le M \} \subset [0,1]$ for a fixed precision parameter $M$, counting hits where $y \le f(x)$, and returning the fraction of hits as the estimator. It is classified as \emph{numerical} in~\cite{brassard:1988} because randomness computes an approximate value. States contain $N$, $G$, $i$, $\mathit{counter}$, and $\mathit{approx}$, with $i=0$, $\mathit{counter}=0$, $\mathit{approx}=\mathit{undef}$ initially.
\begin{align*}
&\textbf{if } i < N \;\textbf{then} \\
& \quad \textbf{choose } x : x \in G \;\textbf{with weight } 1
  \;\textbf{do} \\
& \quad\quad \textbf{choose } y : y \in G \;\textbf{with weight } 1
  \;\textbf{do} \\
& \quad\quad\quad \textbf{if } y \le f(x) \;\textbf{then }
  \mathit{counter} := \mathit{counter} + 1 \\
& \quad i := i + 1 \\
&\textbf{else } \mathit{approx} := \mathit{counter} / N
\end{align*}
\noindent
\textbf{Las Vegas.} The \emph{$n$-queens problem} asks for a placement of $n$ queens on an $n \times n$ board so that no two attack each other. A randomized solver places queens row by row, choosing uniformly among legal columns; if none exist, the attempt fails. Correctness is guaranteed, but runtime is random. States contain integer $i$ (initially $1$) for the current row, a function $\mathit{col} \colon \{1,\dots,n\} \to \{1,\dots,n\} \cup \{\mathit{undef}\}$ for queen positions, and a Boolean $\mathit{success}$ (initially~$\mathbf{true}$):
\begin{align*}
&\textbf{if } i \le n \wedge \mathit{success} \;\textbf{then}\\
&\quad \textbf{let } \mathit{ValidCols} = \{ c \in \{1,\dots,n\} \mid \forall
  j < i:\, c \neq \mathit{col}(j) \wedge |c - \mathit{col}(j)| \neq i - j \}
  \textbf{ in}\\
&\qquad \textbf{if } \mathit{ValidCols} \neq \emptyset \;\textbf{then}\\
&\qquad\quad \textbf{choose } c \in \mathit{ValidCols} \;\textbf{with weight} \; 1 \;\textbf{do} \; \mathit{col}(i) := c \; ; \; i := i + 1\\
&\qquad \textbf{else } \mathit{success} := \textbf{false}
\end{align*}


\noindent
\textbf{Monte Carlo.} The \emph{Miller--Rabin primality test} determines
whether an odd integer $n > 2$ is prime by selecting random bases $a \in
\{2,\dots,n-2\}$ and verifying modular congruences that hold for all integers
coprime to $n$ when $n$ is prime. If any test fails, $n$ is composite;
otherwise $n$ is declared \emph{probably prime}, with error probability
reducible by increasing the number of iterations $k$. States contain $n$, $k$, $i$ (initially $0$),
and $\mathit{probablyPrime}$ (initially $\mathbf{true}$), where $n-1 = 2^s d$
with $d$ odd.
\begin{align*}
&\textbf{if } i < k \wedge \mathit{probablyPrime} \;\textbf{then} \\
& \quad \textbf{choose } a : a \in \{2, \dots, n-2\} \;\textbf{with weight} \; 1 \;\textbf{do} \\
& \quad\quad \textbf{if } (a^d \bmod n) \neq 1 \wedge \forall r < s:\,
  (a^{2^r d} \bmod n) \neq n-1 \;\textbf{then} \; \mathit{probablyPrime} := \textbf{false} \\
& \quad i := i + 1 
\end{align*}
\noindent
\textbf{Sherwood.} \emph{Randomized Quicksort} guarantees correctness but
uses randomness to improve expected performance by selecting pivots uniformly at random. States contain a function $A \colon \{1,\dots,n\} \to B$ representing the array to be sorted and a set $\mathit{Tasks}$ of pairs $(l,h)$ representing unsorted subarrays $A[l\ldots h]$, initially
$\{(1,|A|)\}$. The algorithm terminates when $\mathit{Tasks}$ becomes empty.
\begin{align*}
&\textbf{let } \mathit{Active} = \{(l,h) \in \mathit{Tasks} \mid l < h\}
  \textbf{ in} \\
& \quad \textbf{if } \mathit{Active} \neq \emptyset \;\textbf{then} \\
& \quad \quad \textbf{choose } (l,h) \in \mathit{Active} \;\textbf{with weight
  } 1 \;\textbf{do} \\
& \quad \quad\quad \textbf{choose } p \in \{l,\dots,h\} \;\textbf{with weight
  } 1 \;\textbf{do} \\
& \quad \quad\quad\quad \textbf{let } (A',p') = \mathit{Partition}(A,l,h,p)
  \textbf{ in} \\
& \quad \quad \quad\quad\quad A := A' \; ; \; \mathit{Tasks} := (\mathit{Tasks} \setminus
  \{(l,h)\}) \cup \{(l,p'-1),(p'+1,h)\} \\
&\quad \textbf{else if } \mathit{Tasks} \neq \emptyset \;\textbf{then }
  \mathit{Tasks} := \emptyset
\end{align*}
Here $\mathit{Partition}(A,l,h,p)$ rearranges $A[l\ldots h]$ around pivot
$A[p]$, returning updated array $A'$ and new pivot index $p'$ such that
$A'[l\ldots p'-1] \leq A'[p'] \leq A'[p'+1\ldots h]$. Clearing trivial tasks with $l \geq h$ ensures termination.

\subsection{Probabilistic Algorithms Defined by pASMs}

The postulates in Section~\ref{sec:postulates} define the class of PAs
axiomatically, independently of any machine model. The preceding subsections have defined pASMs as a concrete machine model. The following theorem---proven in Appendix~\ref{app:pasm}---establishes that every pASM is a PA. 

\begin{theorem}[Plausibility]\label{thm-pasm}
Every pASM satisfies the defining postulates for probabilistic algorithms and therefore defines a PA.
\end{theorem}

\section{Probabilistic ASM Thesis}\label{sec:thesis}

We will now address the capture of PAs by pASMs, i.e. we will show the converse of Theorem \ref{thm-pasm}. As the full proof of Theorem \ref{thm-capture} is rather lengthy, it is outsourced to Appendix \ref{app:capture}.

\begin{theorem}[Capture]\label{thm-capture}
For every probabilistic algorithm $\mathcal{A}$ there exists a behaviourally equivalent pASM $M$.
\end{theorem}
\vspace{-1em}
\begin{proof}[Proof sketch]
We fix a bounded exploration witness $W \cup P$ for $\mathcal{A}$, and without loss of generality we can assume that it is closed under subterms. 

Consider a single state $S$ with base set $B$, the set of update sets $\boldsymbol{\Delta}_{\mathcal{A}}(S)$, and the probability measure $d_S$ on it. For $\alpha \in W$ the evaluation in $S$ yields a multiset of multisets $U$, where $U$ contains tuples $(b_0, \dots, b_n)$ of values $b_i \in B$. Each such value $b_i$ is called a {\em $W$-critical value} in $S$, and a {\em $W$-critical tuple} is a tuple of $W$-critical values. Analogously, the evaluation of a p-witness term $\varrho \in P$ in $S$ yields a multiset of tuples $(b_1, \dots, b_m, p)$ of values $b_i \in B$ and $p \in [0,1]$. The values $b_i$ are called {\em $P$-critical}, and a tuple of such values is called a {\em $P$-critical tuple}.

First we can show that whenever $((f,(a_1, \dots, a_n)), a_0)$ is an update in some $\Delta \in \boldsymbol{\Delta}_{\mathcal{A}}(S)$, then $\bar{a} = (a_0, \dots, a_n)$ is a $W$-critical tuple.

Furthermore, the slicing condition in Postulate \ref{p-bew} ensures that there exists a set $W^\prime$ of witness terms such that for the corresponding probabilistic algorithm $\mathcal{A}^\prime$ we get $\boldsymbol{\Delta}_{\mathcal{A}^\prime}(S) = \{ \Delta \}$. If $\alpha = \langle U(y_1,\dots,y_q) \mid \psi(y_1,\dots,y_q) \rangle \in W$ is sliced by $\alpha^\prime = \langle U(y_1,\dots,y_q) \mid \psi(y_1,\dots,y_q) \wedge \chi(y_1,\dots,y_q) \rangle \in W^\prime$, then we get a $q$-tuple $\bar{b}_\alpha$, such that the concatenation $\bar{b}$ of all these tuples for $\alpha \in W$ such that $\chi$ is needed characterises the reduction to $\Delta$ as the only update set. We call $\bar{b}$ a {\em characterising tuple} for the update set $\Delta \in \boldsymbol{\Delta}_{\mathcal{A}}(S)$. As $\psi$ and $\chi$ have the same free variables, characterising tuples are $P$-critical. Furthermore, $\bar{b}$ is associated with a probability $d(\bar{b}) \in [0,1]$, which is the product of the probabilities $p$ resulting from the sliced witness terms.

We exploit the related behavioural theory of parallel algorithms \cite{ferrarotti:tcs2016}, which uses sets $W$ of multiset comprehension terms and critical tuples. Note that in the following lemmata, $\mathcal{A}$ denotes a \emph{parallel} algorithm in the sense of \cite{ferrarotti:tcs2016}, i.e.\ a deterministic algorithm with a state transition function; they are applied below to the substeps $\mathcal{A}^c$ and $\mathcal{A}^p$ of the probabilistic algorithm, not to the probabilistic algorithm itself.


\begin{lemma}[{\cite[Lemma~7.5]{ferrarotti:tcs2016}}]\label{lem-indistinguishable-update}
Let $\mathcal{A}$ be a parallel algorithm, $S$ a state of $\mathcal{A}$, and let $\bar{a} = (a_0, \ldots, a_r)$, and let $(f,(a_1, \ldots, a_r), a_0) \in \Delta_{\mathcal{A}}(S)$. For every $(r+1)$-tuple of critical values $\bar{b} = (b_0, \ldots, b_r)$ with $\mathit{tp}^{\mathrm{FO}_{wo=}}_{S|_{W}}(\bar{b}) = \mathit{tp}^{\mathrm{FO}_{wo=}}_{S|_{W}}(\bar{a})$ the update $(f, (b_1, \ldots, b_r), b_0)$ also belongs to $\Delta_{\mathcal{A}}(S)$.
\end{lemma}

Here $S|_{W}$ denotes a relational structure defined by the bounded exploration witness $W$ of $\mathcal{A}$, and $\mathit{tp}^{\mathrm{FO}_{wo=}}_{S|_{W}}(\bar{a})$ denotes the {\em type} of the tuple $\bar{a}$ in first-order logic without equality with respect to the structure $S|_{W}$. In a nutshell, for every critical tuple $\bar{a}$ that defines an update in some update set, any other critical tuple $\bar{b}$ with the same type also defines an update (for the same function symbol $f$) in the very same update set.

\begin{lemma}[{\cite[Lemma~7.6]{ferrarotti:tcs2016}}]\label{lem-isolating-formula}
For  every relational vocabulary $\Sigma$ with no constants, every finite structure $S$ of schema $\Sigma$, every $r \geq 0$ and every $r$-tuple $\bar{a}$ over $S$ there is a formula $\chi \in \mathit{tp}_S^{\mathrm{FO}_{wo=}}(\bar{a})$  such that for any finite relational structure $S^\prime$ of schema $\Sigma$ and every $r$-tuple $\bar{b}$ over $S^\prime$ we have $S^\prime \models \chi[\bar{b}]$ iff $\mathit{tp}_S^{\mathrm{FO}_{wo=}}(\bar{a}) = \mathit{tp}_{S^\prime}^{\mathrm{FO}_{wo=}}(\bar{b})$.
\end{lemma}

The formula $\chi$ in Lemma \ref{lem-isolating-formula} is called an {\em isolating formula} of the type $\mathit{tp}_S^{\mathrm{FO}_{wo=}}(\bar{a})$. Given a formula $\chi$ which isolates the $\mathrm{FO}_{wo=}$-type of a critical tuple $\bar{a}$ in a critical structure $S|_W$, we can write an equivalent term $t_\chi$ which evaluates to true in $S$ only for those tuples which have the same $\mathrm{FO}_{wo=}$-type as $\bar{a}$ in $S|_W$.

\begin{lemma}[{\cite[Lemma~7.7]{ferrarotti:tcs2016}}]\label{lem-isolating-term}
Let $S$ be a state of a parallel algorithm $\mathcal{A}$ of signature $\Sigma$ with a bounded exploration witness $W$ and let $\bar{a}$ be an $r$-tuple in $(S|_W)^r$ and $\chi$ an isolating formula for the $\mathrm{FO}_{wo=}$-type of $\bar{a}$ in $S|_W$. Then there is a term $t_\chi$ over $\Sigma$ such that for every $\bar{b} \in (S|_W)^r$ we have $\mathrm{val}_{S,\mu[\bar{x}\mapsto\bar{b}]}(t_\chi) = \texttt{true}$ iff $S|_W \models \chi[\bar{b}]$.
\end{lemma}

The term $t_\chi$ in Lemma \ref{lem-isolating-term} is called an {\em isolating term} of the type $\mathit{tp}_S^{\mathrm{FO}_{wo=}}(\bar{a})$. With isolating terms we easily obtain a rule $r_S$ with $\Delta_{r_S}(S) = \Delta_{\mathcal{A}}(S)$. The rule $r_S$ is a bounded parallel combination of parallel rules of the form
\[ \textbf{forall } x_0, x_1, \ldots, x_r : t^{\bar{a}}_{\chi}(x_0, x_1, \ldots, x_r) \textbf{ do } f(x_1, \ldots, x_r) := x_0 \]
for all $\bar{a} = (a_0, a_1, \ldots, a_r) \in (S|_{W})^{r+1}$ with $(f, (a_1, \ldots, a_r), a_0) \in \Delta_{\mathcal{A}}(S)$, where the term $t^{\bar{a}}_{\chi}(x_0, x_1, \ldots, x_r)$ is an isolating term of the critical tuple $\bar{a}$.

\begin{lemma}[{\cite[Cor.~7.8]{ferrarotti:tcs2016}}]\label{lem-rule}
If $S$ is a state of the parallel algorithm $\mathcal{A}$ and $W$ is a bounded exploration witness for $\mathcal{A}$, then $\Delta_{r_S}(S) = \Delta_{\mathcal{A}}(S)$.
\end{lemma}

These lemmata \ref{lem-indistinguishable-update}--\ref{lem-rule} only depend on the bounded exploration witness $W$, i.e. a set of multiset comprehension terms, and a critical tuple $\bar{a}$ for it. This generality is what we need to exploit for probabilistic algorithms.

If $S$ is a state of a probabilistic algorithm $\mathcal{A}$ and $W \cup P$ is
a bounded exploration witness for $\mathcal{A}$, then for any $\Delta \in
\boldsymbol{\Delta}_{\mathcal{A}}(S)$ we obtain a characterising tuple $\bar{b}$
with associated value $d(\bar{b})$, which is a $P$-critical tuple. We can think
of the step of the algorithm in state $S$ as being composed of two substeps. In
the first substep the different characterising tuples for the possible successor
states are determined in parallel and stored in some temporary locations; this
corresponds to a parallel algorithm $\mathcal{A}^c$ with bounded exploration
witness $P$, such that the characterising tuples $\bar{b}$ are $P$-critical
tuples. In the second substep we can use modified locations: if $\bar{a}$ is a critical tuple with respect to $W$ that defines an update with function symbol $f$ in $\Delta \in \boldsymbol{\Delta}_{\mathcal{A}}(S)$ and $\bar{b}$ characterises $\Delta$, then we use a new function symbol $\tilde{f}$ such that $\bar{b}$ is included in the locations used in $\Delta$. This second step (for fixed $\bar{b}$) corresponds to a parallel algorithm $\mathcal{A}^p$ with bounded exploration witness
\begin{align*}
W^p &= \Big\{ \Big\langle (t_1, \dots, t_n, y_1, \dots, y_p) \mid \varphi(x_1, \dots, x_k) \wedge \psi(y_1, \dots, y_q) \Big \rangle \mid \\
&\qquad\qquad \langle \langle (t_1, \dots, t_n) \mid \varphi(x_1, \dots, x_k) \rangle \mid \psi(y_1, \dots, y_q) \rangle \in W \Big\} 
\end{align*}
and the concatenated tuple $\bar{a}\bar{b}$ is $W^p$-critical. 
Applying Lemma \ref{lem-isolating-term} to $\mathcal{A}^c$ and $\mathcal{A}^p$, respectively, gives isolating terms $t_\chi^{\bar{b}}$ and $t_\chi^{\bar{a}\bar{b}}$. With these we define a rule
\begin{align*}
\textbf{choose}\; & y_1, \dots, y_p : t_\chi^{\bar{b}}(y_1, \dots,y_q, p(y_1, \dots,y_q)) \;\textbf{with weight}\; d(\bar{b}) \\
\textbf{do}\; & \textbf{forall}\; x_0, \dots, x_k : t_\chi^{\bar{a}\bar{b}}(x_1,\dots,x_k, y_1, \dots,y_q) \;\textbf{do}\; f(x_1,\dots,x_k) := x_0
\end{align*}
Analogously to Lemma \ref{lem-rule} we then show that the parallel composition of these rules for all characterising tuples $\bar{b}$ and all update-defining tuples in the corresponding update sets defines a rule $r_S$ with $\boldsymbol{\Delta}_{r_S}(S) = \boldsymbol{\Delta}_{\mathcal{A}}(S)$ and $d_{r_S,S} = d_S^{\mathcal{A}}$.

Now we have pASM rules $r_S$ for every state $S$ of the given PA $\mathcal{A}$, which yield the set of update sets and the associated probability measures in that state. However, there are infinitely many states $S \in \mathcal{S}$ and we need a single pASM rule to obtain a behaviourally equivalent pASM.

Therefore, we exploit again the bounded exploration witness to obtain an equivalence relation on $\mathcal{S}$ with only finitely many equivalence classes such that only the rules for representatives of each equivalence class are needed. For a state $S$ define an equivalence relation $\sim_S$ on the bounded exploration witness $W \cup P$ by $\alpha \sim_S \alpha^\prime$ iff $\text{val}_S(\alpha) = \text{val}_S(\alpha^\prime)$. Then call two states $S$, $S^\prime$ {\em $\mathit{WP}$-similar} iff $\sim_S = \sim_{S^\prime}$ holds. As $\sim_S$ defines a partition of $W \cup P$, there can only be finitely many different equivalence relations, and consequently the number of $\mathit{WP}$-similarity classes is also finite. Then show the following lemma.

\begin{lemma}[$\mathit{WP}$-similarity]\label{lem-wp-similarity}
If $\mathcal{A}$ is a probabilistic algorithm with bounded exploration witness $W \cup P$ and $S$, $S^\prime$ are $\mathit{WP}$-similar states, then $\boldsymbol{\Delta}_{r_S}(S^\prime) = \boldsymbol{\Delta}_{\mathcal{A}}(S^\prime)$ and $d_{r_S,S^\prime} = d_{S^\prime}^{\mathcal{A}}$.
\end{lemma}

The proof of this lemma is done in the same way as the proofs of corresponding lemmata in the sequential and parallel ASM theses.

Finally, as there are only finitely many $\mathit{WP}$-similarity classes, we choose representatives $S_1, \dots, S_r$ for these. Then for $1 \le i \le r$ there exists a condition $\varphi_i$ satisfied in $S_i$, but not in the other states $S_j$, so we can define a pASM rule $r$ as a bounded parallel composition of conditional rules \textbf{if} $\varphi_i$ \textbf{then} $r_{S_i}$. It is then easy to show that for this rule we get $\boldsymbol{\Delta}_r(S) = \boldsymbol{\Delta}_{\mathcal{A}}(S)$ and $d_{r,S} = d_S^{\mathcal{A}}$  for all states $S$, which completes the proof.
\end{proof}

\vspace*{-0.5cm}
\section{Concluding Remarks}\label{sec:schluss}

In this paper we motivated and defined four postulates that define the class of probabilistic algorithms (PAs), for which choices according to some probability distribution are essential. We then showed that this class of algorithms is captured by probabilistic Abstract State Machines (pASMs), which constitutes a behavioural theory of PAs.

However, probabilities in our theory are only associated with choices, whereas there is no randomness in the states of a PA. Therefore, a possible extension is to consider also randomised states. Furthermore, we assumed that the set of possible successor states is always finite, so it is a natural question, if this restriction can be removed. 




\bibliography{pasm}

\newpage

\appendix

\section{Probabilistic Algorithms Defined by pASMs}\label{app:pasm}

As stated in Section \ref{sec:pasm} this appendix contains the full proof of Theorem \ref{thm-pasm}.

\begin{theorem}[Plausibility]
Every pASM $M$ satisfies the defining postulates for probabilistic algorithms and therefore defines a PA.
\end{theorem}

\begin{proof}

The requirements of the abstract state and background postulates are built into
the definition of pASMs by construction, so Postulates~\ref{p-state}
and~\ref{p-background} are immediately satisfied. In particular, the sets
$\mathcal{S}$ and $\mathcal{I}$ of states and initial states of $M$ are
well-defined.


For Postulate random branching time, consider an arbitrary state $S$ of $M$, the set $\boldsymbol{\Delta}(S)$ of update sets, and the probability measure $d_{r,S}$ defined on it by the pASM rule semantics. The set of possible successor states of $S$ is $\{ S + \Delta \mid \Delta \in \boldsymbol{\Delta}(S) \}$. We define
$d_S(S + \Delta) = d_{r,S}(\Delta)$ and extend $d_S$ to be $0$ on all other
states. Since the normalisation by $\Sigma_{S,\sigma}$ in the semantics of the
probabilistic choice rule (Section~\ref{ssec:rules}) ensures that $d_{r,S}$
sums to $1$ over $\boldsymbol{\Delta}(S)$, it follows that $d_S$ is a proper
probability measure on $\mathcal{S}$, satisfying Postulate~\ref{p-rbt}.


It remains to show the satisfaction of Postulate~\ref{p-bew} (probabilistic bounded exploration). For this we construct a standard bounded exploration witness $W_r \cup P_r$ from the rule $r$ of a given pASM $M$. The construction of witness terms $\alpha \in W_r$ is done by structural recursion, and the probability defining terms in probabilistic choice rules define the p-witness terms $\varrho \in P_r$ as follows:

\begin{enumerate}\renewcommand{\labelenumi}{(\arabic{enumi})}

\item If $r$ is a \textbf{skip} rule, then we simply have $W_r = \{ \langle \langle () \mid \textbf{true} \rangle \mid \textbf{true} \rangle \}$ and $P_r = \{ \langle (1) \mid \textbf{true} \rangle \}$.

\item If $r$ is an assignment rule $f(t_1,\dots,t_n) := t_0$, then we have $P_r = \{ \langle (1) \mid \textbf{true} \rangle \}$ and
$W_r = \{ \langle \langle (t_0, \dots, t_n) \mid \textbf{true} \rangle \mid \textbf{true} \rangle \}$.

\item If $r$ is a conditional rule \textbf{if} $\varphi$ \textbf{then} $r_1$ \textbf{else} $r_2$, we have $P_r = P_{r_1} \cup P_{r_2}$ and
\begin{align*}
W_r &= \Big \{ \langle \langle \varphi \mid \textbf{true} \rangle \mid \textbf{true} \rangle \Big \} \cup
\Big \{ \langle \langle t \mid \psi \wedge \varphi \rangle \mid \chi \rangle \mid \langle \langle t \mid \psi \rangle \mid \chi \rangle \in W_{r_1} \Big \} \\
&\qquad\qquad \cup \Big \{ \langle \langle t \mid \psi \wedge \neg\varphi \rangle \mid \chi \rangle \mid \langle \langle t \mid \psi \rangle \mid \chi \rangle \in W_{r_2} \Big \} \; .
\end{align*}

\item If $r$ is a parallel rule \textbf{forall} $x_1 ,\dots, x_k : \varphi$ \textbf{do} $r^\prime(x_1 ,\dots, x_k)$, we have $P_r = P_{r^\prime}$ and
\begin{align*}
W_r &= \Big \{ \langle \langle t \mid \psi \wedge \varphi(x_1, \dots, x_k) \rangle \mid \chi \rangle \mid \langle \langle t \mid \psi \rangle \mid \chi \rangle \in W_{r^\prime(x_1,\dots,x_k)} \Big \} \\
&\qquad\qquad \cup \Big \{ \langle \langle \varphi(x_1,\dots,x_k) \mid \textbf{true} \rangle \mid \textbf{true} \rangle \Big \} \; .
\end{align*}

\item If $r$ is a choice rule \textbf{choose} $x_1 ,\dots, x_k : \varphi$ \textbf{with weight} $p(x_1 ,\dots, x_k)$ \textbf{do} $r^\prime$, we have
\begin{align*}
W_r &= \Big \{ \langle \langle t \mid \psi \rangle \mid \chi \wedge \varphi(x_1, \dots, x_k) \rangle \mid \langle \langle t \mid \psi \rangle \mid \chi \rangle \in W_{r^\prime(x_1,\dots,x_k)} \Big \} \\
&\qquad \cup \Big \{ \langle \langle \varphi(x_1,\dots,x_k) \mid \textbf{true} \rangle \mid \textbf{true} \rangle \Big \} \; \text{and} \\
P_r &= \{ \langle (x_1,\dots,x_k, y_1,\dots,y_q, \tilde{p}(x_1,\dots,x_k) \cdot p^\prime(y_1,\dots,y_q)) \mid \\
&\qquad \psi(y_1,\dots,y_q) \wedge \varphi(x_1,\dots,x_k) \rangle \mid \langle (y_1,\dots,y_q, p^\prime(y_1,\dots,y_q)) \mid \\
&\qquad\qquad \psi(y_1,\dots,y_q) \rangle \in P_{r^\prime(x_1,\dots,x_k)} \} \; ,
\end{align*}
where $\tilde{p}(x_1,\dots,x_k) = \cfrac{p(x_1,\dots,x_k)}{\sum p(x_1,\dots,x_k)}$.
\end{enumerate}

The alignment of $W$ with $P$ is obvious. 
We first show that states that coincide on $W \cup P$ yield the same sets of update sets with the same probability distribution on them. For this we use structural induction on the rule $r$, so let $S_1$ and $S_2$ be two states of $M$ and let $\zeta: V \rightarrow B$ be a variable assignment on the set $V$ of free variables of a rule $r$. Assume that with $\zeta$ the states $S_1, S_2$ coincide on $W \cup P$, where $W$ and $P$ are defined as above for the rule $r$.

{\bf (1)} Let $r$ be a \textbf{skip} rule. Then we always have $\boldsymbol{\Delta}_{r,\zeta}(S_1) = \{ \emptyset \} = \boldsymbol{\Delta}_{r,\zeta}(S_2)$. As $\emptyset$ is the only update set, we further have $d_{S_i}(\emptyset) = 1$.

{\bf (2)} If $r$ is an assignment rule $f(t_1,\dots,t_n) := t_0$ with $V = \bigcup_{0 \le i \le n} fr(t_i)$, we get 
\[ \text{val}_{S_i,\zeta}(\langle \langle (t_0, \dots, t_n) \mid \textbf{true} \rangle \mid \textbf{true} \rangle) =
\langle \langle (\text{val}_{S_i,\zeta}(t_0), \dots, \text{val}_{S_i,\zeta}(t_n)) \rangle \rangle \; . \]

Consider the locations $\ell_i = (f, (\text{val}_{S_i,\zeta}(t_1), \dots, \text{val}_{S_i,\zeta}(t_n)))$ and the values $b_i = \text{val}_{S_i,\zeta}(t_0)$. As $S_1$ and $S_2$ coincide on $W_r$ using $\zeta$, we have $\ell_1 = \ell_2$ and $b_1 = b_2$, hence $\boldsymbol{\Delta}_{r,\zeta}(S_1) = \{ \{ (\ell_1, b_1) \} \} = \{ \{ (\ell_2, b_2) \} \} = \boldsymbol{\Delta}_{r,\zeta}(S_2)$. In addition, as there is only a single update set $\Delta$, we have $d_{S_i}(\Delta) = 1$.

{\bf (3)} Let $r$ be a conditional rule \textbf{if} $\varphi$ \textbf{then} $r_1$ \textbf{else} $r_2$, and let $V_0 = fr(\varphi)$, $V_1 = fr(r_1)$ and $V_2 = fr(r_2)$. Let $\zeta: V_0 \cup V_1 \cup V_2 \rightarrow B$ be a variable assignment. As $\langle\langle \varphi \mid \textbf{true} \rangle \mid \textbf{true} \rangle \in W_r$, we have $val_{S_i,\zeta}(\varphi) = \textbf{true}$ for both $i \in \{ 1,2 \}$ or $val_{S_i,\zeta}(\varphi) = \textbf{false}$ for both $i \in \{ 1,2 \}$.

Condider the first case. Then with $\zeta$ $S_1$, $S_2$ coincide on $\langle\langle t \mid \psi \wedge \varphi \rangle \mid \chi \rangle$ for all $\langle\langle t \mid \psi \rangle \mid \chi \rangle \in W_{r_1}$, hence they also coincide on $W_{r_1}$. By induction $\boldsymbol{\Delta}_{r_1,\zeta}(S_1) = \boldsymbol{\Delta}_{r_1,\zeta}(S_2)$, and by definition $\boldsymbol{\Delta}_{r,\zeta}(S_i) = \boldsymbol{\Delta}_{r_1,\zeta}(S_i)$ holds, from which the claim follows immediately.

Analogously, in the second case $S_1$, $S_2$ coincide on $\langle\langle t \mid \psi \wedge \neg\varphi \rangle \mid \chi \rangle$ for all $\langle\langle t \mid \psi \rangle \mid \chi \rangle \in W_{r_2}$, hence they also coincide on $W_{r_2}$. By induction $\boldsymbol{\Delta}_{r_2,\zeta}(S_1) = \boldsymbol{\Delta}_{r_2,\zeta}(S_2)$, and by definition $\boldsymbol{\Delta}_{r,\zeta}(S_i) = \boldsymbol{\Delta}_{r_2,\zeta}(S_i)$ holds.

In addition, we assume that $S_1$ and $S_2$ coincide on $P_r = P_{r_1} \cup P_{r_2}$, hence they also coincide on $P_{r_1}$. By induction we get $\boldsymbol{\Delta}_{r_1,\zeta}(S_1) = \boldsymbol{\Delta}_{r_1,\zeta}(S_2)$ and $d_{r_1,S_1,\zeta} = d_{r_1,S_2,\zeta}$.
Altogether we get $\boldsymbol{\Delta}_{r,\zeta}(S_1) = \boldsymbol{\Delta}_{r,\zeta}(S_2)$ and $d_{r,S_1,\zeta} = d_{r,S_2,\zeta}$.


{\bf (4)} Let $r$ be a parallel rule \textbf{forall} $x_1 ,\dots, x_k : \varphi$ \textbf{do} $r^\prime(x_1 ,\dots, x_k)$, and for $V_0 = fr(t)$ and $V = fr(r(v))$ let $\zeta: (V \cup V_0) - \{ v \} \rightarrow B$ be a variable assignment. As $\langle \langle \varphi(x_1,\dots,x_k) \mid \textbf{true} \rangle \mid \textbf{true} \rangle \in W_r$, we have that
\[ B = \{ (a_1,\dots,a_k) \mid val_{S_i,\zeta[x_1 \mapsto a_1, \dots, x_k \mapsto a_k]}(\varphi) = \textbf{true} \} \]
is the same set of tuples for both $i \in \{ 1,2 \}$.

As $S_1$, $S_2$ with $\zeta$ coincide on $\langle \langle t \mid \psi \wedge \varphi(x_1, \dots, x_k) \rangle \mid \chi \rangle$ for all $\langle \langle t \mid \psi \rangle \mid \chi \rangle \in W_{r^\prime(x_1,\dots,x_k)}$, they also coincide on $W_{r^\prime(a_1,\dots,a_k)}$ for all $(a_1,\dots,a_k) \in B$. Hence by induction we have $\boldsymbol{\Delta}_{r^\prime(a_1,\dots,a_k)}(S_1) = \boldsymbol{\Delta}_{r^\prime(a_1,\dots,a_k)}(S_2)$ for all $(a_1,\dots,a_k) \in B$. By definition 
\[ \boldsymbol{\Delta}_{r,\zeta}(S_i) = \Big \{ \bigcup_{(a_1,\dots,a_k) \in B} \Delta_{(a_1,\dots,a_k)} \mid  \Delta_{(a_1,\dots,a_k)} \in \boldsymbol{\Delta}_{r^\prime(a_1,\dots,a_k)}(S_i) \Big \} \]
hold for $i \in \{ 1,2 \}$. Furthermore, as $S_1$ and $S_2$ coincide on $P_r$, they also coincide on $P_{r^\prime(x_1, \dots, x_k)}$. By induction we get $d_{r^\prime(x_1,\dots,x_k),S_1,\zeta} = d_{r^\prime(x_1,\dots,x_k),S_2,\zeta}$, which implies $d_{r,S_1,\zeta} = d_{r,S_2,\zeta}$.

{\bf (5)} Let $r$ be a choice rule \textbf{choose} $x_1 ,\dots, x_k : \varphi$ \textbf{with weight} $p(x_1 ,\dots, x_k)$ \textbf{do} $r^\prime$ with $V_0 = fr(t)$, $V = fr(r(v))$, and let $\zeta: (V \cup V_0) - \{ v \} \rightarrow B$ be a variable assignment. As in the previous case we have $\langle \langle \varphi(x_1,\dots,x_k) \mid \textbf{true} \rangle \mid \textbf{true} \rangle \in W_r$, and hence the set
\[ B = \{ (a_1,\dots,a_k) \mid val_{S_i,\zeta[x_1 \mapsto a_1, \dots, x_k \mapsto a_k]}(\varphi) = \textbf{true} \} \]
of tuples is the same for both $i \in \{ 1,2 \}$.

As $S_1$, $S_2$ with $\zeta$ coincide on $\langle \langle t \mid \psi \rangle \mid \chi \wedge \varphi(x_1, \dots, x_k) \rangle$ for all $\langle \langle t \mid \psi \rangle \mid \chi \rangle \in W_{r^\prime(x_1,\dots,x_k)}$, they also coincide on $W_{r^\prime(a_1,\dots,a_k)}$ for all $(a_1,\dots,a_k) \in B$. By induction we get $\boldsymbol{\Delta}_{r^\prime(a_1,\dots,a_k)}(S_1) = \boldsymbol{\Delta}_{r^\prime(a_1,\dots,a_k)}(S_2)$ for all $(a_1,\dots,a_k) \in B$. By definition we get
\[ \boldsymbol{\Delta}_{r,\zeta}(S_i) = \bigcup_{(a_1,\dots,a_k) \in B} \boldsymbol{\Delta}_{r^\prime(a_1,\dots,a_k)}(S_i) \]
for $i \in \{ 1,2 \}$.

By induction we further get $d_{r^\prime(x_1,\dots,x_k),S_1,\zeta} = d_{r^\prime(x_1,\dots,x_k),S_2,\zeta}$, which implies $d_{r,S_1,\zeta} = d_{r,S_2,\zeta}$.




Concerning the slicing condition in Postulate \ref{p-bew} any slicing of a witness term $\alpha \in W_r$ corresponds to a strengthening of a probabilistic choice subrule by replacing the condition $\psi(y_1, \dots, y_q)$ by $\psi(y_1, \dots, y_q) \wedge \chi(y_1, \dots, y_q)$. If the probability term $p(y_1, \dots, y_q)$ in this rule is replaced by $p(y_1, \dots, y_q) / \sum_{y_1^\prime, \dots, y_q^\prime} p(y_1^\prime, \dots, y_q^\prime)$, where the sum ranges over those $(y_1^\prime, \dots, y_q^\prime)$ for which $\psi(y_1^\prime, \dots, y_q^\prime) \wedge \chi(y_1^\prime, \dots, y_q^\prime)$ holds, then we obtain a modified pASM rule $r^\prime$, and the pASM defined by this rule yields sets of update sets $\boldsymbol{\Delta}_{r^\prime}(S) \subseteq \boldsymbol{\Delta}_r(S)$.
\end{proof}

\section{The Capture of Probabilistic Algorithms by pASMs}\label{app:capture}

As stated in Section \ref{sec:thesis} this appendix contains the full proof of Theorem \ref{thm-capture}. 

\begin{theorem}[Capture]
For every probabilistic algorithm $\mathcal{A}$ there exists a behaviourally equivalent pASM $M$.
\end{theorem}

For this we want to exploit as much as possible the related behavioural theory of parallel algorithms \cite{ferrarotti:tcs2016}, so we summarise a few essential results without repeating the partly lengthy proofs.

A {\em parallel algorithm} is defined by for postulates: 

\begin{description}

\item[Sequential Time.] A parallel algorithm $\mathcal{A}$ has a set $\mathcal{S}$ of states, a set $\mathcal{I} \subseteq \mathcal{S}$ of initial states, and a state transition function $\tau: \mathcal{S} \rightarrow \mathcal{S}$.

\item[Abstract State.] There exists a signature $\Upsilon$ such that all states $S \in \mathcal{S}$ are $\Upsilon$-structures. The sets $\mathcal{S}$ and $\mathcal{I}$ are closed under isomorphisms. The successor state $\tau(S)$ of every state $S$ has the same base set as $S$. Any isomorphism between states $S$ and $S^\prime$ is also an isomorphism between $\tau(S)$ and $\tau(S^\prime)$.

\item[Background.] The background requirements are the same as in Postulate \ref{p-background}.

\item[Bounded Exploration.] There exists a finite set $W$ of multiset comprehension terms of the form $\alpha = \langle (t_1, \dots, t_n) \mid \varphi \rangle$ such that for all states $S$, $S^\prime$ that coincide on $W$ we have $\Delta_{\mathcal{A}}(S) = \Delta_{\mathcal{A}}(S^\prime)$, where $\Delta_{\mathcal{A}}(S)$ is the unique minimal consistent update set with $\tau(S) = S + \Delta(S)$ (see Lemma \ref{lem-update-set}).

\end{description}

If we fix a {\em bounded exploration witness} $W$, assuming without loss of generality that $W$ is closed under subterms, then for every state $S$ and $\alpha \in W$ we obtain a multiset $\text{val}_S(\alpha)$ of tuples of values of the base set of $S$. Any $b \in B$ that appears in one such tuple is called a {\em critical value} in state $S$. A tuple of such critical values is called a {\em critical tuple}. The following sequence of lemmata was proven in \cite{ferrarotti:tcs2016}.

\begin{lemma}[{\cite[Lemma~7.1]{ferrarotti:tcs2016}}]
For every state $S$, if $((f,(a_1, \dots, a_n)), a_0)$ is an update in $\Delta_{\mathcal{A}}(S)$, then $\bar{a} = (a_0, \dots, a_n)$ is a critical tuple in $S$.
\end{lemma}

\begin{lemma}[{\cite[Lemma~7.5]{ferrarotti:tcs2016}}]
Let $\mathcal{A}$ be a parallel algorithm, $S$ a state of $\mathcal{A}$, and let $\bar{a} = (a_0, \ldots, a_r)$, and let $(f,(a_1, \ldots, a_r), a_0) \in \Delta_{\mathcal{A}}(S)$. For every $(r+1)$-tuple of critical values $\bar{b} = (b_0, \ldots, b_r)$ with $\mathit{tp}^{\mathrm{FO}_{wo=}}_{S|_{W}}(\bar{b}) = \mathit{tp}^{\mathrm{FO}_{wo=}}_{S|_{W}}(\bar{a})$ the update $(f, (b_1, \ldots, b_r), b_0)$ also belongs to $\Delta_{\mathcal{A}}(S)$.
\end{lemma}

Here $S|_{W}$ denotes a relational structure defined by the bounded exploration witness $W$ of $\mathcal{A}$. The vocabulary of $S|_{W}$ is $\Sigma_W = \{ R_{\alpha_1}, \dots, R_{\alpha_m} \}$ for $W = \{ \alpha_1, \dots, \alpha_m \}$. For $\alpha_i = \langle (t_0, \dots, t_n) \mid \varphi_i(x_1, \dots, x_r) \rangle$ the relation symbol $R_{\alpha_i}$ has arity $n+2$ and the following interpretation:
\begin{align*}
R_{\alpha_i}^{S|_{W}} &= \{(b_0, \ldots, b_{n_i}, j \cdot b_0 \cdots b_{n_i}  \cdot a_1 \cdots a_{r_i}) \mid  (a_1, \ldots, a_{r_i}) \in B^{r_i}, \\
&\quad {\bf S} \models \varphi_i(x_1, \ldots, x_{r_i})[a_1, \ldots, a_{r_i}] \text{ and }\\
&\quad \text{val}_{S,\mu[x_1 \mapsto a_1, \ldots, x_{r_i} \mapsto a_{r_i}]}(t_0) = b_0, \ldots, \text{val}_{{\bf S},\mu[x_1 \mapsto a_1, \ldots, x_{r_i} \mapsto a_{r_i}]}(t_{n_i}) = b_{n_i}\} \; , 
\end{align*}
where $j \cdot b_0 \cdots b_{n_i}  \cdot a_1 \cdots a_{r_i}$ denotes a value in $B$ such that the set of these values for fixed $b_0, \dots, b_{n_i}$ is in one-one correspondence to the set of tuples $(a_1, \dots, a_{r_i})$ satisfying the required conditions for the evaluation of $\varphi_i$ and $t_k$ ($0 \le k \le n_i$). An element $a_i$ belongs to the domain of $S|_{W}$ iff for some $\alpha_i \in W$ there is a $\bar{a} \in R^{S|_{W}}_{\alpha_i}(\bar{a})$ such that $a_i$ appears in $\bar{a}$.

Furthermore, $\mathit{tp}^{\mathrm{FO}_{wo=}}_{S|_{W}}(\bar{a})$ denotes the {\em type} of the tuple $\bar{a}$ in first-order logic without equality with respect to the structure $S|_{W}$, i.e. the set of all formulae satisfied by $\bar{a}$ in the structure $S|_{W}$.

In a nutshell, the lemma states that for every critical tuple $\bar{a}$ that defines an update in some update set any other critical tuple $\bar{b}$ with the same type also defines an update (for the same function symbol $f$) in the very same update set.

\begin{lemma}[{\cite[Lemma~7.6]{ferrarotti:tcs2016}}]
For  every relational vocabulary $\Sigma$ with no constants, every finite structure $S$ of schema $\Sigma$, every $r \geq 0$ and every $r$-tuple $\bar{a}$ over $S$ there is a formula $\chi \in \mathit{tp}_S^{\mathrm{FO}_{wo=}}(\bar{a})$  such that for any finite relational structure $S^\prime$ of schema $\Sigma$ and every $r$-tuple $\bar{b}$ over $S^\prime$ we have $S^\prime \models \chi[\bar{b}]$ iff $\mathit{tp}_S^{\mathrm{FO}_{wo=}}(\bar{a}) = \mathit{tp}_{S^\prime}^{\mathrm{FO}_{wo=}}(\bar{b})$.
\end{lemma}

The formula $\chi$ in Lemma \ref{lem-isolating-formula} is called an {\em isolating formula} of the type $\mathit{tp}_S^{\mathrm{FO}_{wo=}}(\bar{a})$. Given a formula $\chi$ which isolates the $\mathrm{FO}_{wo=}$-type of a critical tuple $\bar{a}$ in a critical structure $S|_W$, we can write an equivalent term $t_\chi$ which evaluates to true in $S$ only for those tuples which have the same $\mathrm{FO}_{wo=}$-type as $\bar{a}$ in $S|_W$.

\begin{lemma}[{\cite[Lemma~7.7]{ferrarotti:tcs2016}}]
Let $S$ be a state of a parallel algorithm $\mathcal{A}$ of signature $\Sigma$ with a bounded exploration witness $W$ and let $\bar{a}$ be an $r$-tuple in $(S|_W)^r$ and $\chi$ an isolating formula for the $\mathrm{FO}_{wo=}$-type of $\bar{a}$ in $S|_W$. Then there is a term $t_\chi$ over $\Sigma$ such that for every $\bar{b} \in (S|_W)^r$ we have
\[\mathrm{val}_{S,\mu[\bar{x}\mapsto\bar{b}]}(t_\chi) = \texttt{true} \quad \text{iff} \quad S|_W \models \chi[\bar{b}].\]
\end{lemma}

The term $t_\chi$ in Lemma \ref{lem-isolating-term} is called an {\em isolating term} of the type $\mathit{tp}_S^{\mathrm{FO}_{wo=}}(\bar{a})$. With isolating terms we easily obtain a rule $r_S$ with $\Delta_{r_S}(S) = \Delta_{\mathcal{A}}(S)$. The rule $r_S$ is a bounded parallel combination of parallel rules of the form
\[ \textbf{forall } x_0, x_1, \ldots, x_r : t^{\bar{a}}_{\chi}(x_0, x_1, \ldots, x_r) \textbf{ do } f(x_1, \ldots, x_r) := x_0 \]
for all $\bar{a} = (a_0, a_1, \ldots, a_r) \in (S|_{W})^{r+1}$ with $(f, (a_1, \ldots, a_r), a_0) \in \Delta_{\mathcal{A}}(S)$, where $t^{\bar{a}}_{\chi}(x_0, x_1, \ldots, x_r)$ is an isolating term of the critical tuple $\bar{a}$.

\begin{lemma}[{\cite[Cor.~7.8]{ferrarotti:tcs2016}}]
If $S$ is a state of the parallel algorithm $\mathcal{A}$ and $W$ is a bounded exploration witness for $\mathcal{A}$, then $\Delta_{r_S}(S) = \Delta_{\mathcal{A}}(S)$.
\end{lemma}

We will exploit that these lemmata only depend on the bounded exploration witness $W$, i.e. a set of multiset comprehension terms, and a critical tuple $\bar{a}$ for it. This generality is what we need for probabilistic algorithms. 

From now on let $\mathcal{A}$ denote again a probalistic algorithm. We fix a bounded exploration witness $W \cup P$ for $\mathcal{A}$, and without loss of generality we can assume that it is closed under subterms. 

First consider a single state $S$ with base set $B$, the set of update sets $\boldsymbol{\Delta}_{\mathcal{A}}(S)$, and the probability measure $d_S$ on it. For $\alpha \in W$ the evaluation in $S$ yields a multiset of multisets $U$, where $U$ contains tuples $(b_0, \dots, b_n)$ of values $b_i \in B$. Each such value $b_i$ is called a {\em $W$-critical value} in $S$, and a {\em $W$-critical tuple} is a tuple of $W$-critical values. Analogously, the evaluation of a p-witness term $\varrho \in P$ in $S$ yields a multiset of tuples $(b_1, \dots, b_m, p)$ of values $b_i \in B$ and $p \in [0,1]$. The values $b_i$ are called {\em $P$-critical}, and a tuple of such values is called a {\em $P$-critical tuple}.

First we can show that the values appearing in the updates in one of the update sets in $\boldsymbol{\Delta}_{\mathcal{A}}(S)$ are again critical values.

\begin{lemma}\label{lem-w-critical}
For every state $S$, if $((f,(a_1, \dots, a_n)), a_0)$ is an update in $\Delta \in \boldsymbol{\Delta}_{\mathcal{A}}(S)$, then $\bar{a} = (a_0, \dots, a_n)$ is a $W$-critical tuple in $S$.
\end{lemma}

\begin{proof}
Assume that $a_i$ is not $W$-critical. Then choose a structure $S_1$ by replacing $a_i$ by a fresh value $b$ without changing anything else. Then $S_1$ is isomorphic to $S$ and thus a state by the Postulate \ref{p-state}.

Let $\alpha \in W$ be a critical term. Then we must have $\text{val}_S(\alpha) = \text{val}_{S_1}(\alpha)$, hence $S$ and $S_1$ coincide on $W$. Postulate \ref{p-bew} implies $\boldsymbol{\Delta}_{\mathcal{A}}(S) = \boldsymbol{\Delta}_{\mathcal{A}}(S_1)$ and hence $(f(a_1,\dots,a_n),a_0) \in \Delta \in \boldsymbol{\Delta}_{\mathcal{A}}(S_1)$. However, $a_i$ does not appear in the structure $S_1$ and thus cannot appear in $\Delta \in \boldsymbol{\Delta}_{\mathcal{A}}(S_1)$, which gives a contradiction.
\end{proof}

Furthermore, the slicing condition in Postulate \ref{p-bew} ensures that there exists a set $W^\prime$ of witness terms such that for the corresponding probabilistic algorithm $\mathcal{A}^\prime$ we get $\boldsymbol{\Delta}_{\mathcal{A}^\prime}(S) = \{ \Delta \}$. If $\alpha = \langle U(y_1,\dots,y_q) \mid \psi(y_1,\dots,y_q) \rangle \in W$ is sliced by $\alpha^\prime = \langle U(y_1,\dots,y_q) \mid \psi(y_1,\dots,y_q) \wedge \chi(y_1,\dots,y_q) \rangle \in W^\prime$, then we get a $q$-tuple $\bar{b}_\alpha$, such that the concatenation $\bar{b}$ of all these tuples for $\alpha \in W$ such that $\chi$ is needed characterises the reduction to $\Delta$ as the only update set. We call $\bar{b}$ a {\em characterising tuple} for the update set $\Delta \in \boldsymbol{\Delta}_{\mathcal{A}}(S)$. As $\psi$ and $\chi$ have the same free variables, characterising tuples are $P$-critical. Furthermore, $\bar{b}$ is associated with a probability $d(\bar{b}) \in [0,1]$, which is the product of the probabilities $p$ resulting from the sliced witness terms.

Then we can think of the step of the algorithm in state $S$ being the sequence of two substeps. In the first step the different characterising tuples for the possible successor states are determined in parallel and stored in some temporary locations. Then in a second step one of these characterising tuples is chosen, and the corresponding update set is built and applied. 

The first step corresponds to a parallel algorithm $\mathcal{A}^c$ with bounded exploration witness $P$, such that the characterising tuples $\bar{b}$ of any update set are $P$-critical tuples. In the second step we can use modified locations: if $\bar{a}$ is a critical tuple with respect to $W$ that defines an update with function symbol $f$ in $\Delta \in \boldsymbol{\Delta}_{\mathcal{A}}(S)$ and $\bar{b}$ characterises $\Delta$, then we use a new function symbol $\tilde{f}$ such that $\bar{b}$ is included in the locations used in $\Delta$. This second step (for fixed $\bar{b}$) corresponds to a parallel algorithm $\mathcal{A}^p$ with bounded exploration witness
\begin{align*}
W^p &= \bigg\{ \Big\langle (t_1, \dots, t_n, y_1, \dots, y_p) \mid \varphi(x_1, \dots, x_k) \wedge \psi(y_1, \dots, y_q) \Big \rangle \mid \\
&\qquad \langle \langle (t_1, \dots, t_n) \mid \varphi(x_1, \dots, x_k) \rangle \mid \psi(y_1, \dots, y_q) \rangle \in W \bigg\} 
\end{align*}
and the concatenated tuple $\bar{a}\bar{b}$ is $W^p$-critical. We then can apply the results for parallel algorithms, which give us isolating terms $t_\chi^{\bar{b}}$ and $t_\chi^{\bar{a}\bar{b}}$ for $\mathcal{A}^c$ and $\mathcal{A}^p$, respectively. With these we can define a rule
\begin{align}
& \textbf{choose}\; y_1, \dots, y_p : t_\chi^{\bar{b}}(y_1, \dots,y_q, p(y_1, \dots,y_q)) \;\textbf{with weight}\; d(\bar{b}) \notag\\
&\quad \textbf{do}\; \textbf{forall}\; x_0, \dots, x_k : t_\chi^{\bar{a}\bar{b}}(x_1,\dots,x_k, y_1, \dots,y_q) \;\textbf{do}\; f(x_1,\dots,x_k) := x_0 \label{eq-rule}
\end{align}

With these we define a rule $r_S$ as the parallel composition of all rules as in (\ref{eq-rule}) for all characterising tuples $\bar{b}$ and all update-defining tuples in the corresponding update sets. 

\begin{lemma}\label{lem-pa-rule}
If $S$ is a state of the probabilistic algorithm $\mathcal{A}$, $W$ is a bounded exploration witness for $\mathcal{A}$ and $r_S$ is the rule above with respect to $W$, then we get $\boldsymbol{\Delta}_{r_S}(S) = \boldsymbol{\Delta}_{\mathcal{A}}(S)$ and $d_{r_S,S} = d_S^{\mathcal{A}}$.
\end{lemma}

\begin{proof}
By definition, one update set $\Delta \in \boldsymbol{\Delta}_{r^S}(S)$ is characterised by $\bar{b}$, and $((f,(a_1,\dots,a_k)),a_0)$ is an update in $\Delta$. Then due to the construction of isolating tuples the parallel rule yields also all updates $((f,(a_1^\prime,\dots,a_k^\prime)),a_0^\prime)$, where the type of $\bar{a}^\prime\bar{b}$ (for $\bar{a}^\prime = (a_0^\prime, \dots, a_k^\prime)$) is the same as the type of $\bar{a}^\prime\bar{b}$ (with respect to $W^p$). As $r_S$ is the parallel composition of rules as defined by (\ref{eq-rule}). the rule $r_S$ yields the complete update set $\Delta$.

Analogously, due to the construction of isolating terms, the probabilistic choice rule yields all characterising tuples $\bar{b}$ for the update sets in $\boldsymbol{\Delta}_{\mathcal{A}}(S)$ together with their probabilities $d(\bar{b})$, which implies $\boldsymbol{\Delta}_{r^S}(S) = \boldsymbol{\Delta}_{\mathcal{A}}(S)$ and $d_{r_S,S} = d_S^{\mathcal{A}}$ as claimed.
\end{proof}

Now we have pASM rules $r_S$ for every state $S$ of the given PA $\mathcal{A}$, which yield the set of update sets and the associated probability measures in that state. However, there are infinitely many states $S \in \mathcal{S}$ and we need a single pASM rule to obtain a behaviourally equivalent pASM.

Therefore, we exploit again the bounded exploration witness to obtain an equivalence relation on $\mathcal{S}$ with only finitely many equivalence classes such that only the rules for representatives of each equivalence class are needed. 

For a state $S$ define an equivalence relation $\sim_S$ on the bounded exploration witness $W \cup P$ by $\alpha \sim_S \alpha^\prime$ iff $\text{val}_S(\alpha) = \text{val}_S(\alpha^\prime)$. Then call two states $S$, $S^\prime$ {\em $\mathit{WP}$-similar} iff $\sim_S = \sim_{S^\prime}$ holds. As $\sim_S$ defines a partition of $W \cup P$, there can only be finitely many different equivalence relations, and consequently the number of $\mathit{WP}$-similarity classes is also finite. Then we show the following three lemmata.

\begin{lemma}\label{lem-pa-isomorphism}
Let $S$ be a state of $\mathcal{A}$ and let $S_1, S_2$ be isomorphic states of $\mathcal{A}$. If $\boldsymbol{\Delta}_{r_S}(S_1) = \boldsymbol{\Delta}_{\mathcal{A}}(S_1)$ and $d_{r_S,S_1} = d_{S_1}^{\mathcal{A}}$ hold, then also $\boldsymbol{\Delta}_{r_S}(S_2) = \boldsymbol{\Delta}_{\mathcal{A}}(S_2)$ and $d_{r_S,S_2} = d_{S_2}^{\mathcal{A}}$ hold.
\end{lemma}

\begin{proof}
Let $\sigma$ be an isomorphism with $\sigma(S_1) = S_2$. As $\sigma$ is a bijection between the base sets $B_1$ and $B_2$ of the states $S_1$ and $S_2$, respectively, it can be extended to locations, updates, update sets and also sets of update sets, hence $\sigma(\boldsymbol{\Delta}_{\mathcal{A}}(S_1)) = \boldsymbol{\Delta}_{\mathcal{A}}(S_2)$ holds, and $d_{r_S,S_1}(S^\prime) = d_{r_S,S_2}(\sigma(S^\prime))$ holds for all states $S^\prime$ with base set $B$. This immediately implies the claim.
\end{proof}

\begin{lemma}\label{lem-pa-coincidence}
If $S, S^\prime$ are states of the probabilistic algorithm $\mathcal{A}$, which coincide on the bounded exploration witness $W \cup P$, then we get $\boldsymbol{\Delta}_{r_S}(S^\prime) = \boldsymbol{\Delta}_{\mathcal{A}}(S^\prime)$ and $d_{r_S,S^\prime} = d_{S^\prime}^{\mathcal{A}}$.
\end{lemma}

\begin{proof}
By Postulate \ref{p-bew} we have $\boldsymbol{\Delta}_{\mathcal{A}}(S) = \boldsymbol{\Delta}_{\mathcal{A}}(S^\prime)$ and $d_S^{\mathcal{A}} = d_{S^\prime}^{\mathcal{A}}$. With Lemma \ref{lem-pa-rule} we get $\boldsymbol{\Delta}_{r_{S^\prime}}(S^\prime) = \boldsymbol{\Delta}_{\mathcal{A}}(S^\prime)$ and $d_{r_{S^\prime},S^\prime} = d_{S^\prime}^{\mathcal{A}}$.

Furthermore, the relational structures $S \mid_W$ and $S^\prime$ must be isomorphic, and the isomorphism $\zeta$ maps critical values to critical values. Therefore, for every update $((f,(a_1,\dots,a_n),a_0)) \in \Delta \in \boldsymbol{\Delta}(S)$ the type $\mathit{tp}^{\mathrm{FO}_{wo=}}_{S|_{W^p}}(\bar{a}\bar{b})$ coincides with the type $\mathit{tp}^{\mathrm{FO}_{wo=}}_{S^\prime|_{W^p}}(\bar{a}\bar{b})$, where $\bar{b}$ is a characterising tuple for $\Delta$ and the same holds for the type of $\bar{b}$ with respect to structures $S|_P$ and $S^\prime|_P$.
Hence $r_S = r_{S^\prime}$ holds by construction of the rule $r_S$, and we get $\boldsymbol{\Delta}_{r_S}(S^\prime) = \boldsymbol{\Delta}_{r_{S^\prime}}(S^\prime)$. Combining these equalities implies the claim.
\end{proof}

\begin{lemma}[$\mathit{WP}$-similarity]
If $\mathcal{A}$ is a probabilistic algorithm with bounded exploration witness $W \cup P$ and $S$, $S^\prime$ are $\mathit{WP}$-similar states, then $\boldsymbol{\Delta}_{r_S}(S^\prime) = \boldsymbol{\Delta}_{\mathcal{A}}(S^\prime)$ and $d_{r_S,S^\prime} = d_{S^\prime}^{\mathcal{A}}$.
\end{lemma}

\begin{proof}
Without loss of generality we can assume that the base sets $B, B^\prime$ of $S$ and $S^\prime$, respectively, are disjoint. Otherwise, we can replace $S^\prime$ by an isomorphic copy $\sigma(S^\prime)$, prove the lemma for this state instead, so we get $\boldsymbol{\Delta}_{r_S}(\sigma(S^\prime)) = \boldsymbol{\Delta}_{\mathcal{A}}(\sigma(S^\prime))$ and $d_{r_S,\sigma(S^\prime)} = d_{\sigma(S^\prime)}^{\mathcal{A}}$, and apply Lemma \ref{lem-pa-isomorphism} to show the claim.

Now due to Lemma \ref{lem-pa-isomorphism} it suffices to show the claim for a state $\bar{S}$ that is isomorphic to $S^\prime$. So we construct an isomorphism $\zeta$ with $\zeta(\text{val}_{S^\prime}(\alpha) = \text{val}_S(\alpha)$ for all $\alpha \in W \cup P$, and consider $\bar{S} = \zeta(S^\prime)$. This is possible, because $S$ and $S^\prime$ are disjoint and $WP$-similar.

Then by construction $S$ and $\bar{S}$ coincide on $W \cup P$, so we get $\boldsymbol{\Delta}_{r_S}(\bar{S}) = \boldsymbol{\Delta}_{\mathcal{A}}(\bar{S})$ and $d_{r_S,\bar{S}} = d_{\bar{S}}^{\mathcal{A}}$ by Lemma \ref{lem-pa-coincidence}, which completes the proof.
\end{proof}

Finally, as there are only finitely many $\mathit{WP}$-similarity classes, we choose representatives $S_1, \dots, S_r$ for these. Then for $1 \le i \le r$ there exists a condition $\varphi_i$, which is satisfied in $S_i$, but not in the other states $S_j$, so we can define a pASM rule $r$ as a bounded parallel composition of conditional rules \textbf{if} $\varphi_i$ \textbf{then} $r_{S_i}$. Then it is rather straightforward to show that for this rule we get $\boldsymbol{\Delta}_(S) = \boldsymbol{\Delta}_{\mathcal{A}}(S)$ and $d_{r,S} = d_S^{\mathcal{A}}$, which completes the proof of Theorem \ref{thm-capture}.

\end{document}